\documentclass[a4paper,UKenglish,cleveref,autoref,thm-restate]{lipics-v2021}
\usepackage{xcolor}
\usepackage{macros}
\usepackage{tikz}
\usetikzlibrary{petri,positioning,automata, arrows.meta}
\usepackage[nocompress]{cite}
\usepackage{multicol}
\usepackage{graphicx}
\RequirePackage{amsfonts, amsmath, amssymb}
\usepackage{proof}
\usepackage{todonotes}
\usepackage{hyperref}
\usepackage{tikz}
\usepackage{bm}
\usetikzlibrary{petri,positioning,automata, arrows.meta}
\usepackage{esvect}
\usepackage{mathtools}
\usepackage{enumerate}
\usepackage{wasysym}
\usepackage{mathpartir} %
\usepackage{multicol}
\usepackage{wrapfig}
\usepackage{caption}

\makeatletter
\newcommand\notsotiny{\@setfontsize\notsotiny\@viipt\@viiipt}
\makeatother

\usepackage{apptools}

\AtAppendix{\counterwithin{lemma}{section}}

\title{Unreliability in Practical Subclasses of Communicating Systems}

\author{Amrita Suresh}
{University of Oxford, UK}{
amrita.suresh@cs.ox.ac.uk 
}{%
https://orcid.org/0000-0001-6819-9093 
}{}

\author{Nobuko Yoshida}
{University of Oxford, UK}{
nobuko.yoshida@cs.ox.ac.uk 
}{%
https://orcid.org/0000-0002-3925-8557
}{}

\authorrunning{A. Suresh and N. Yoshida} %

\Copyright{Amrita Suresh and Nobuko Yoshida}



\relatedversion{} 

\funding{This work is supported by EPSRC EP/T006544/2, EP/T014709/2, EP/Y005244/1, EP/V000462/1, EP/X015955/1, 
EP/Z0005801/1, EU Horizon (TARDIS) 101093006 and ARIA}

\acknowledgements{We thank Martin Vassor for his comments on an early version of this paper.}%

\nolinenumbers %

\EventEditors{C. Aiswarya, Ruta Mehta, and Subhajit Roy}
\EventNoEds{3}
\EventLongTitle{45th IARCS Annual Conference on Foundations of Software Technology and Theoretical Computer Science (FSTTCS 2025)}
\EventShortTitle{FSTTCS 2025}
\EventAcronym{FSTTCS}
\EventYear{2025}
\EventDate{December 17--19, 2025}
\EventLocation{BITS Pilani, K K Birla Goa Campus, India}
\EventLogo{}
\SeriesVolume{360}
\ArticleNo{14}

\begin{document}

\maketitle

\begin{abstract}
Systems of communicating automata are
prominent 
models for peer-to-peer message-passing over unbounded channels,
but in the general scenario, most verification properties are
undecidable. To address this issue,
two decidable subclasses, \emph{Realisable with Synchronous
  Communication} ($\RSC$) 
and \emph{$k$-Multiparty Compatibility} ($\kmc$),
were proposed in the literature, with corresponding verification tools developed and applied in practice. Unfortunately, both 
$\RSC$ and $\kmc$ are not resilient under failures: (1) their decidability relies on
the assumption of perfect channels and (2) most standard
protocols do not satisfy
$\RSC$ or $\kmc$ under failures. 
To address these limitations, 
this paper studies the resilience of 
$\RSC$ and $\kmc$ under 
two distinct failure models: \emph{interference} and \emph{crash-stop
  failures}.
For interference,
we relax the conditions of $\RSC$
and $\kmc$ 
and prove that the inclusions of  
these relaxed properties
remain decidable under interference, preserving their known complexity bounds.  
We then propose a novel crash-handling communicating system that
captures wider behaviours than existing multiparty session types
(MPST) with crash-stop failures.  We study a translation of MPST with
crash-stop failures into this system integrating $\RSC$ and $\kmc$
properties, and establish their decidability results. Finally, by verifying representative protocols from the literature using $\RSC$ and $\kmc$ tools extended to interferences, we evaluate the relaxed systems and demonstrate their resilience.
\end{abstract}

\section{Introduction}
\label{sec:introduction}
Asynchronous processes that communicate using First In First Out
(FIFO) channels \cite{brand_communicating_1983}, henceforth referred to as FIFO systems, are widely 
used to model distributed protocols, 
but their verification is known to be computationally challenging. 
The model is Turing-powerful for even just two processes communicating via two unidirectional FIFO channels 
\cite{brand_communicating_1983}.

To address this challenge, several efforts have focused on identifying practical yet decidable subclasses -- those expressive enough to model a wide range of distributed protocols, while ensuring that verification problems such as reachability and model checking remain decidable.
Most FIFO systems assume \emph{perfect} channels, which is too restrictive 
to model the real-world distributed phenomena where
system failures often happen. 
This paper investigates whether two practical decidable subclasses of communicating systems,
\emph{Realisable with Synchronous 
  Communication} ($\RSC$) \cite{di_giusto_towards_2021} 
and \emph{$k$-Multiparty Compatibility}
($\kmc$) \cite{lange_verifying_2019}, 
are \emph{resilient}
when integrated with two different kinds of \emph{failures}. These failure models were originally introduced in the contexts of \emph{contracts} \cite{lozes_reliable_2012} and \emph{session types} \cite{barwell_designing_2023, barwell_generalised_2022}.  %
We say a system is \emph{resilient} under a given failure model 
if (i) the inclusion remains decidable, and (ii) the verification properties 
of interest remain decidable under that failure model. 

\subparagraph*{Failures in communications.}
A widely studied failure model in FIFO systems 
is \emph{lossy channels}. Finkel \cite{finkel_decidability_1994} showed that the termination problem is decidable for the class of \emph{completely specified protocols}, a model which strictly includes FIFO systems with lossy channels.
Abdulla and Jonsson~\cite{abdulla_verifying_1993} developed 
algorithms for verifying termination, safety, and eventuality
properties for protocols on lossy channels, by showing that they
belong to the class of well-structured transition systems. %

Another type of failure, studied in a more practical setting, occurs when one or more processes \emph{crash}. 
In the most general case,  Fekete et al.~\cite{fekete_impossibility_1993} proved that 
if an underlying process crashes, no fault-tolerant reliable communication protocol can be implemented. To address this, they consider faultless models which attempt to capture the behaviour of crashes by broadcasting crash messages. Such approaches have been explored in the context of \emph{runtime verification} techniques \cite{kazemlou_crash-resilient_2018} and \emph{session types} \cite{barwell_designing_2023, barwell_generalised_2022,BHYZ2025}. In this work, we closely study a failure model proposed by \cite{barwell_designing_2023,BHYZ2025}. 

{
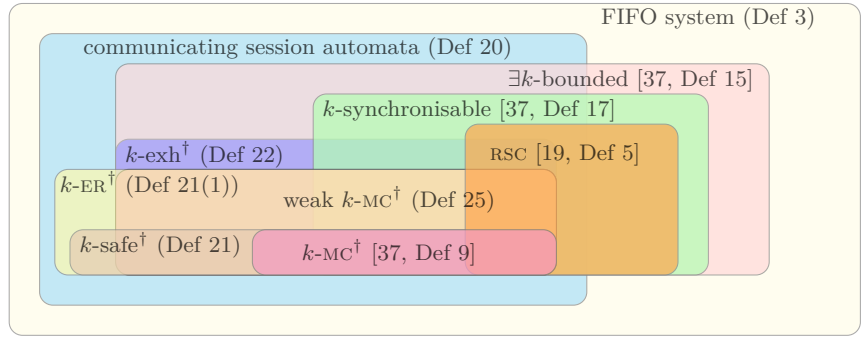
\begin{figure}[t]
	\begin{center}
		\begin{tikzpicture}[scale = 0.4, opacity = 0.7]
			\draw[rounded corners = 5pt, gray!95, fill=yellow!10]
			(-3,-2) rectangle (25,9);
			\draw[rounded corners = 5pt, gray!95, fill=cyan!30]
			(-2,-1) rectangle (16,8);
			\draw[rounded corners = 5pt, gray!95, fill=pink!60]
			(0.5,0) rectangle (22,7); 
			\draw[rounded corners = 5pt, gray!75, fill=blue!40]
			(0.5,0) rectangle (15,4.5);    
			\draw[rounded corners = 5pt, gray!90, fill=green!30]
			(7,0) rectangle (20,6);
			\draw[rounded corners = 5pt, black, fill=magenta!60]
			(0.5,0) rectangle (15,3.5);
			\draw[rounded corners = 5pt, gray!90, fill=yellow!40]
			(-1.5,0) rectangle (15,3.5);
			\draw[rounded corners = 5pt, gray!90, fill=brown!40]
			(-1,0) rectangle (15,1.5);
			\draw[rounded corners = 5pt, gray!95, fill=orange!70]
			(12,0) rectangle (19,5);
			\draw[rounded corners = 5pt, gray, fill=magenta!50]
			(5,0) rectangle (15,1.5);
			\draw (20, 8.5) node {\footnotesize{\text{FIFO system}}      \footnotesize{(Def~\ref{def:ca})}} ;
			\draw (6.5, 7.5) node {\footnotesize{\text{communicating session automata}} \footnotesize{(Def~\ref{def:csa})}} ;
			\draw (7, 5.45) node[right]{\footnotesize{\text{$k$-synchronisable} \cite[Def~17]{lange_verifying_2019}}};
			\draw (21.8, 6.45) node[left] {\footnotesize{\text{$\exists$$k$-bounded} \cite[Def~15]{lange_verifying_2019}}} ;
			\draw (0.5, 4) node[right] {\footnotesize{\text{$k$-exh}$^\dagger$ (Def~\ref{def:kexh})}};
			\draw (-1.65, 3) node[right] {\footnotesize{\text{$k$-\textsc{er}}$^\dagger$ (Def~\ref{def:ksafe}(1))}};
			\draw (9.5, 2.5) node
                              {\footnotesize{\text{weak
                                    $k$-\textsc{mc}}$^\dagger$ (Def~\ref{def:wkmc})}};
			\draw (9.5, 0.75) node {\footnotesize{\text{$k$-\textsc{mc}}$^\dagger$ \cite[Def~9]{lange_verifying_2019}}};
			\draw (-1.0, 1) node[right] {\footnotesize{\text{$k$-safe}$^\dagger$ (Def~\ref{def:ksafe})}};
			\draw (12.5, 4) node[right] {\footnotesize{\text{\textsc{rsc}} \cite[Def~5]{di_giusto_multiparty_2023}}};
		\end{tikzpicture}
\caption{Classes of communication systems (since the $\dagger$-marked definitions are introduced in the context of CSA (Def \ref{def:csa}), we restrict them accordingly).}
\label{fig:diagram_comparison}
\end{center}
\end{figure}
 }
\subparagraph*{Restricting the channel behaviour.}
To define decidable subclasses,  
many works 
study how to restrict read and write access to channels.  For two-process (binary) FIFO systems, the notion of \emph{half-duplex} communication was introduced in \cite{cece_verification_2005}, where at most one direction of communication is active at any time. For such systems, reachability is decidable in polynomial time. However, generalising this idea to the multiparty setting often yields subclasses that are either too restrictive or lose decidability.  
Di Giusto et al.~\cite{di_giusto_multiparty_2023, di_giusto_towards_2021} extended this idea to multiparty systems while preserving decidability, resulting in the notion of systems \emph{realisable with synchronous communication} ($\RSC$).
They showed that this definition overlaps with that in
\cite{cece_verification_2005} for mailbox communication. However, in the
case of peer-to-peer communication where the two definitions differ,
peer-to-peer $\RSC$ behaviour was
proved to be decidable. $\RSC$ systems are related to \emph{synchronisable systems} \cite{bouajjani_completeness_2018, di_giusto_k-synchronizability_2020, bollig_unifying_2021}, in which FIFO behaviours must admit a synchronisable execution. The tool \emph{ReSCu} applies this idea to verify real-world distributed protocols \cite{desgeorges_rsc_2023}. %

Another approach to restricting channel behaviours is to bound the length of the channel. Lohrey~\cite{lohrey_bounded_2002} introduced 
\emph{existentially bounded systems} (see also \cite{genest_kleene_2004,genest_communicating_2007}) where all executions that reach a final state with empty channels can be re-ordered into a bounded execution. Although many verification problems are decidable for this class of systems, checking if a system is existentially $k$-bounded is undecidable, even if $k$ is given as part of the input.
	
A decidable bounded approach, \emph{$k$-multiparty
  compatibility} ($\kmc$), 
was introduced in \cite{lange_verifying_2019}. 
This property is 
defined by two conditions, \emph{exhaustivity} and \emph{safety}. Exhaustivity implies existential boundedness and characterises systems where each automaton behaves the same way under bounds of a certain length.
Checking $\kmc$ is decidable, and 
the tool $\kmc$-checker is implemented and applied to verify 
Rust \cite{cutner_deadlock-free_2022,LNY2022} and OCaml \cite{ImaiLN22} programs. 

\subparagraph*{Combining the two approaches.}  
As far as we are aware, the intersection between expressive, decidable subclasses and communication failures %
is less explored. 
Lozes and Villard \cite{lozes_reliable_2012} studied reliability in binary half-duplex systems and showed that 
many communicating contracts can be verified with this model. 
Inspired by this, we investigate whether practical \emph{multiparty} 
subclasses, $\RSC$ and $\kmc$, remain robust in the presence of communication errors.

Although failure models such as lossy channels are well studied, the complexity of verification in their presence is often very high — for instance, reachability in lossy systems is non-primitive recursive \cite{finkel_decidability_1994}. Our goal is not only to show that $\RSC$ and $\kmc$ systems are resilient, but also that their inclusion remains decidable under failure models, with complexity maintained from the failure-free case.

This paper extends $\RSC$
\cite{di_giusto_multiparty_2023,di_giusto_towards_2021} and
$\kmc$ \cite{lange_verifying_2019} by integrating 
two distinct failure models.  
$\RSC$ and $\kmc$ systems are incomparable to each other 
($\RSC$ is not a subset of $\kmc$ and vice-versa),  
but both are closely related to existentially bounded systems. Figure~\ref{fig:diagram_comparison} illustrates their relationship with other models. %

For failures, 
first, we consider \emph{interferences} from the environments
by modelling 
\emph{lossiness}, 
\emph{corruption} (a message is altered to a different message) 
and \emph{out-of-ordering} (two messages in a queue are swapped)
of channels
studied in the context of FIFO systems. %
Secondly, we consider 
potential \emph{crashes} of processes introduced in the setting of 
session types \cite{barwell_designing_2023, barwell_generalised_2022}.

\begin{figure}[t]
\begin{center}
\small{
		\begin{tikzpicture}[>=stealth,node distance=1.1cm,shorten >=1pt,
			every state/.style={text=black, scale =0.45, minimum size=0.1em}, semithick,
			font={\fontsize{8pt}{12}\selectfont}]
			\node[state, initial, initial text=] (place1) at (0,0) {};
			\node[state, right=1.5 cm of place1] (place2) {};
			\node[state, below=of place2] (place3) {};
			\node[state, above=of place1] (place4) {};
			\node[left= 2.5cm of place2] (heading) {Process $\procx{s}$};
			
			\path [-stealth, thick]
			(place2) edge node[right=0.01cm] {$\sendact{\chan{s}{r}}{end}$}(place3)
			(place2) edge[loop right] node[right=0.01cm] {$\sendact{\chan{s}{r}}{data}$} (place2)
			(place1) edge node[above=0.01cm] {$\sendact{\chan{s}{r}}{start}$} (place2)
			(place4) edge[bend left] node[right=0.01cm] {$\sendact{\chan{s}{r}}{data}$} (place1)
			(place1) edge[bend left] node[left=0.01cm] {$\recact{\chan{r}{s}}{err}$} (place4)
			(place3) edge node[left=0.1cm] {$\recact{\chan{r}{s}}{\textit{ack}}$} (place1);

			\node[state, initial, initial text=, right=7cm of place1] (place5) {};
			\node[state, right=1.5 cm of place5] (place6) {};
			\node[state, below=of place6] (place7) {};
			\node[state, above=of place5] (place8) {};
		\node[left=2.5cm of place6] (heading2) {Process $\procr$};
			
			\path [-stealth, thick]
			(place6) edge node[right=0.01cm]{$\recact{\chan{s}{r}}{end}$}(place7)
			(place6) edge[loop right] node[right=0.01cm] {$\recact{\chan{s}{r}}{data}$} (place6)
			(place5) edge node[above=0.01cm] {$\recact{\chan{s}{r}}{start}$} (place6)
			(place5) edge[bend left] node[left=0.01cm] {$\recact{\chan{s}{r}}{data}$} (place8)
			(place8) edge[bend left] node[right=0.01cm] {$\sendact{\chan{r}{s}}{err}$} (place5)
			(place7) edge node[left=0.1cm] {$\sendact{\chan{r}{s}}{ack}$} (place5);
		
	\end{tikzpicture}}
\vspace*{-2mm}
	\caption{The above system $\system$ is half-duplex in the absence of errors. However, in case of (any or multiple) errors, it is no longer half-duplex.}
	\label{fig:ex1}
 \end{center}
\end{figure}
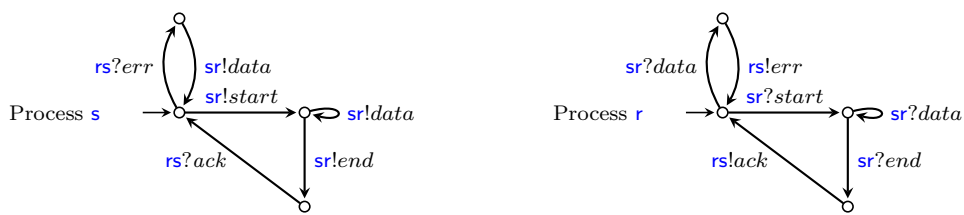

Let us consider the following simple half-duplex protocol 
as an example.

\begin{example}
The system in Figure~\ref{fig:ex1} 
is half-duplex under the assumption of perfect channels 
\cite{cece_verification_2005}.
It consists of two processes, a $\procx{s}$ender ($\procx{s}$) and a
(dual) $\procr$eceiver ($\procr$), communicating via unbounded
FIFO channels. A transition $\sendact{\chan{s}{r}}{m}$ denotes that the
$\procx{s}$ender puts (asynchronously) a message $\msg{m}$ on
channel $\bchan{s}{r}$, and similarly, $\recact{\chan{s}{r}}{m}$
denotes that message $\msg{m}$ is consumed by the $\procr$eceiver
from channel $\bchan{s}{r}$. Since the channel $\bchan{r}{s}$
only contains messages after the
$\procr$eceiver receives the $\msg{end}$ message and has emptied
$\bchan{s}{r}$, the system satisfies the half-duplex condition.
Moreover, the $\procx{s}$ender never sends $\msg{data}$ without having first sent $\msg{start}$ so the error loop is never triggered.

Now suppose the channels are prone to \emph{corruption}. 
A message $\msg{data}$ could be altered to $\msg{end}$ after being sent.
This allows the $\procr$eceiver to react prematurely by sending $\msg{ack}$,
while the $\procx{s}$ender continues sending $\msg{data}$.
As a result, both channels may become non-empty,
violating the half-duplex property.
Similarly
in the presence of other forms of interference, as shown in Example~\ref{ex:interferences},
this system no longer satisfies the half-duplex condition.
\end{example}

\subparagraph*{Contributions and outline.}
The main objective of this paper is to investigate
whether multiparty adaptations of half-duplex systems 
($\RSC$ and $\kmc$) retain both their expressiveness for modelling real-world protocols and the decidability of their inclusions, with preserved complexity, under two distinct kinds of communication failures: interferences and crash-stops. 
\S~\ref{sec:prelim} introduces preliminary notions, notably FIFO
systems and interference models;
\S~\ref{sec:rsc} 
 studies $\RSC$ under interference, and shows that relaxing certain conditions on matching send and receive actions retains both expressiveness and decidability (Theorem~\ref{thm:rscmain}).
\S~\ref{sec:kmc} 
examines $\kmc$ with interferences, and proposes a relaxed version, $\kwmc$ (weak $\kmc$), by weakening the \emph{safety} condition. We prove that checking the $\kwmc$ property remains decidable under interferences (Theorem~\ref{thm:kmcmainthm}).
\S~\ref{sec:crashstop} introduces 
the FIFO systems with crash-stop
failures (called \emph{crash-handling systems}), and shows that
checking $\RSC$ and $\kwmc$ under crash-stop failures 
is decidable (Theorems~\ref{prop:rsccrashdec} and \ref{thm:kmcmain}); 
\S~\ref{sec:crashstopsessions}
defines a translation from (local) multiparty session types (MPST) to
crash-handling systems and 
proves that this translation preserves trace semantics.
This implies the decidability of $\RSC$ and $\kmc$ within
the asynchronous MPST system extended to crash-stop failures 
(Theorem~\ref{lem:fifocrash}); \S~\ref{sec:implementation} evaluates
protocols
from the literature extending the existing tools with support for interferences;
and \S~\ref{sec:conclusion} concludes with further related and future work. 
Proofs are provided in the appendix. 
The tools and benchmarks are publicly available from \url{https://github.com/NobukoYoshida/Interference-Tool}.
\section{Preliminaries}\label{sec:prelim}
For a finite set $\Alphset$, we denote by $\Alphset^\ast$ the set of finite words over $\Alphset$, and the empty word with $\emptyword$. We use $|w|$ to denote the length of the word $w$, and $w_1\cdot w_2$ indicates the concatenation of two words $w_1, w_2 \in \Alphset^\ast$. Given a (non-deterministic) finite-state automaton $\automata$, we denote by $\lang(\automata)$ the language accepted by $\automata$. 
Consider a finite set of processes $\procset$ (ranged over by $\procp, \procq, \procr, \ldots$ or occasionally by $\procx{r_1}, \procx{r_2}, \ldots$) and a set of messages $\msgset$. In this paper, we consider the \emph{peer-to-peer} communication model; i.e., there is a pair of unidirectional FIFO channels between each pair of processes, one for each direction of communication. In our model, processes 
act either by point-to-point communication or by internal actions (actions local to a single process). Moreover, in this setting, we consider messages to be atomic, akin to letters of an alphabet.\\%

Let $\chanset$ $ = \{\chan{\procp}{\procq} \mid \procp \neq \procq \text{ and }\procp, \procq \in \procset\}$ be a set of \emph{channels} that stand for  
point-to-point links.  
Since we are considering the peer-to-peer model of communication, 
there is a unique process that can send a message to (or dually, receive a message from) a particular channel. 
An action $a =
(\chan{\procp}{\procq}, !, \msg{m}) \in \actset$ indicates that a
process $\procp$ sends a message $\msg{m}$ on the channel
$\chan{\procp}{\procq}$. Similarly, $a = (\chan{\procq}{\procp}, ?, \msg{m})
\in \actset$ indicates that $\procp$ receives a message $\msg{m}$ on the
channel $\chan{\procq}{\procp}$. 
We henceforth denote an action $a =
(\chan{\procp}{\procq}, \dagger, \msg{m}) \in \actset$, where $\dagger \in
\{!, ?\}$, in a simplified form as $\chan{\procp}{\procq}{\dagger} \msg{m}$.
An \emph{internal action} $\internalp{p}$ means that process $\procp$ 
performs the action $c$. 
We define a finite set of actions as $\actset \subseteq (\chanset \times \{!, ?\} \times\msgset )
\cup \internalset$ where 
$\internalset$ is the set of all internal actions. 

\begin{definition}[FIFO automaton]
	A FIFO automaton $\automata_{\procp}$, associated with $\procp$, is defined as $\automata_\procp = (\cstateset_\procp, \transrel_\procp, \initp{p})$ where:
		 $\cstateset_\procp$ is the finite set of control-states, 
		$\transrel_\procp \subseteq \cstateset_\procp \times \actset \times \cstateset_\procp$ 
		is the transition relation, and
		$\initp{p} \in \cstateset_\procp$ is the initial control-state.
\end{definition}

Note that in this model, there are no final or accepting states. 

The set of outgoing channels of process $\procp$ 
is represented by $\chansetout =
\{\chan{\procp}{\procq} \mid \procq \in \procset \setminus \procp\}$. 
Similarly,
$\chansetin = \{\chan{\procq}{\procp} \mid \procq \in
\procset \setminus \procp\}$ is the set of incoming channels of
process $\procp$.

Given an action $\act$, 
an \emph{active process}, denoted by $\proc(\act)$, 
is defined as: 
 $\proc(\actsend{\procp}{\procq}{m}) = \procp$ and $\proc(\actrec{\procp}{\procq}{m}) = \procq$. Similarly, $\ch(\actsend{\procp}{\procq}{m}) = \ch(\actrec{\procp}{\procq}{m}) = \chan{\procp}{\procq}$.

We say a control state $q \in \cstateset_\procp$ is a \emph{sending state} (resp. \emph{receiving state}) if all its outgoing transitions are labelled by send (resp. receive) actions. If a control state is neither a sending nor receiving state, i.e., it either has both send and receive actions or neither, then we call it a \emph{mixed state}. We say a sending (resp. receiving) state is \emph{directed} if all the send (resp. receive) actions from that control state are to the same process. 
Like for finite-state automata, 
we say that a FIFO automaton 
$\automata_\procp = (\cstateset_\procp, \transrel_\procp, 
\initp{p})$ is \emph{deterministic} if for all transitions $(q, \act, q'), (q, \act', q'') \in \transrel_\procp$, $\act = \act' \implies q' = q''$. 
We write 
$q_1 \xrightarrow{\act_1 \cdots \act_{n}} q_{n+1}$ for $(q_1, \act_1, q_2) \cdots (q_{n}, \act_{n}, q_{n+1})
\in \transrel_\procp$. 
Unless specified otherwise, we consider non-deterministic automata, 
allowing mixed states, and all states do not have to be directed.

\begin{definition}[FIFO system] \label{def:ca}
A FIFO system $\system = (\automata_\procp)_{\procp \in \procset}$ is a set of communicating FIFO automata. A configuration of $\system$ is a
	pair $\conf = \confq{\globq}{\chancon}$ where $\globq =(q_\procp)_{\procp \in \procset}$ is called the \emph{global state} with $q_\procp \in \cstateset_\procp$ being one of the local control-states of $\automata_{\procp}$, and where $\chancon = (\word_{\chan{\procp}{\procq}})_{\chan{\procp}{\procq \in \chanset}}$ with $\word_{\chan{\procp}{\procq}} \in \msgset^\ast$.
\end{definition}\label{def:fifoconf}

The \emph{initial configuration} of $\system$ is $\initconf =
\confq{\globqinit}{\chanconinit}$ where $\globqinit =
(\initp{p})_{\procp \in \procset}$ and we write $\chanconinit$ for the
$|\chanset|$-tuple $(\emptyword, \ldots, \emptyword)$. We
let $\automata_\procp = (\cstateset_\procp,
\transrel_\procp, \initp{p})$ be a FIFO automaton. Let $\system =
(\automata_\procp)_{\procp \in \procset}$ be the system
whose initial configuration is $\initconf$.
The FIFO automaton $\product(\system)$ associated with $\system$ is
the standard asynchronous product automaton: $\product(\system) =
(\cstateset, \transrel, \globqinit)$ where $\cstateset =
\prod_{\procp \in \procset}\cstateset_\procp$, $\globqinit =
(\initp{p})_{\procp \in \procset}$, and $\transrel$ is the set of
triples $(\globqone, a, \globqtwo)$ for which there exists $\procp \in
\procset$ such that $({q_1}_\procp, a, {q_2}_\procp) \in
\delta_\procp$ and ${q_1}_\procr = {q_2}_\procr$ for all $\procr \in
\procset \setminus \{\procp\}$.
An \emph{execution} $e = a_1 \cdot a_2 \cdots a_n \in \actset^\ast$ is an
arbitrary finite sequence of actions. 
We write $\exec(\system)$ for $\{e \in
\actset^\ast \mid \initconf \xrightarrow{e} \conf \text{ for some
  configuration } \conf\}$. Given $e = a_1 \cdot a_2 \cdots a_n$, we
write $\actset(e) = \{ a_1, a_2, \ldots, a_n\}$. Moreover, we say two systems are trace-equivalent if they produce the same set of executions, i.e. $\system \approx \system'$ is as follows:  $\forall \phi, \phi \in \exec(\system) \Leftrightarrow \phi \in \exec{(\system')}$.

\subparagraph*{Interferences.}
In this paper, we do not restrict
the study to perfect channels, and instead consider that they may
subject to \emph{interferences} from the environment. 
Interferences are modelled as potential evolution of channel contents
without a change in the global state of the system.  As in \cite{lozes_reliable_2012}, we model interferences by a preorder over words ${\succeq} \subseteq \Alphset^* \times \Alphset^*$, which will parametrise the semantics of dialogue systems. We denote by $w \succeq w'$ if $w$ and $w'$ are the contents of the buffer respectively before and after the interferences.

\begin{definition}[Interference model] (from \cite{lozes_reliable_2012})
	An interference model is a binary relation ${\succeq} \subseteq \Alphset^* \times \Alphset^*$ satisfying the following axioms:
\begin{center}
\begin{tabular}{ccccc}
Reflexivity~ & ~Transitivity~ & ~Additivity~ & ~Integrity~ & ~Non-expansion \\ [0.1cm]
$\infer{ a \succeq a}{a \in \Alphset}$ & $\infer{w \succeq w''}{w \succeq w' \;  w' \succeq w''}$ & $\infer{w_1\cdot w_2 \succeq w'_1 \cdot w'_2}{w_1 \succeq w'_1\; w_2 \succeq w'_2}$ & $\infer{w = \varepsilon}{\varepsilon \succeq w}$ & $\infer{|w| \geq |w'|}{w \succeq w'}$
\end{tabular}
\end{center}
 \end{definition}

Axiom \emph{Additivity} defines that failures can happen at any part of
the words; axiom \emph{Integrity} says $\varepsilon$ is 
the least word; and axiom \emph{Non-expansion}
says that $\succeq$ preserves the size of words. 
Based on interferences, we define three failures as follows: 
\begin{itemize}
	\item \emph{Lossiness}:~ Possible leaks of messages during transmission are modelled by adding the axiom $a \succeq \emptyword$.
	\item \emph{Corruption}:~ Possible transformation of a message $a$ into a message $b$ is modelled by adding the axiom $a \succeq b$.
	\item \emph{Out-of-order}:~ Out-of-order communications are modelled by adding axioms $a\cdot b \succeq b \cdot a$ for all $a, b \in \Alphset$.
\end{itemize}
We now define successor configurations for FIFO systems with interferences.
\begin{definition}[Successor configuration under interference]
	\label{def:irs}
	Let $\system$ be a FIFO system. A configuration 
	$\conf' = \confq{\globq'}{\chancon'}$ is a  \emph{successor} of another 
	configuration $\conf = \confq{\globq}{\chancon}$, 
	by firing the transition $(q_\procp, a, q'_\procp) \in \transrel_\procp$,  
	written $\conf \xrightarrow{} \conf'$ or 
	$\conf \xrightarrow{a} \conf'$, 
	if either:
	{\rm (1)} $a= \procp\procq ! m$ 
	and {\rm (a)} $q'_{\procr}=q_{\procr}$ for all $\procr\not = \procp$; 
	and {\rm (b)} $w'_{\procp\procq} \preceq w_{\procp\procq}\cdot m$ and 
	$w'_{\procr\procs} \preceq w_{\procr\procs}$ for all $\procr\procs\neq \procp\procq$; or {\rm (2)} 
	$a= \procq\procp ? m$  
	and {\rm (a)} $q'_{\procr}=q_{\procr}$ for all $\procr\not = \procp$; 
	and {\rm (b)} $m\cdot w_{\procq\procp}' \preceq  w_{\procq\procp}$ and 
	$w'_{\procr\procs} \preceq w_{\procr\procs}$ for all $\procr\procs\neq \procq\procp$. 
\end{definition}

The condition (1-b) puts the content to a channel $\chan{\procp}{\procq}$, 
while (2-b) gets the content
from a channel $\chan{\procp}{\procq}$. 
The reflexive and transitive closure of $\xrightarrow{}$
is $\xrightarrow{\ast}$. 
We write $\conf_1 \xrightarrow{a_1\cdot a_2 \cdots a_m} \conf_{m+1}$
for $\conf_1 \xrightarrow{a_1}\conf_2 \cdots \conf_m \xrightarrow{a_m} \conf_{m+1}$. Moreover, we write $(\conf_1, {a_1\cdot a_2 \cdots a_m}, \conf_{m+1}) \subseteq \transrel$ to denote $\{(\conf_1, {a_1},\conf_2), \ldots, (\conf_m,{a_m},\conf_{m+1})\} \subseteq \transrel$.
A configuration $\conf$ is \emph{reachable} if 
$\conf_0\astred \conf$ and we define 
$\RS{\system}=\{ \conf \ | \ \conf_0\astred \conf\}$. 

A configuration $\conf = \confq{\globq}{\chancon}$ is said to be \emph{$k$-bounded} if for all $\chan{\procp}{\procq} \in \chanset$, $|\word_{\chan{\procp}{\procq}}| \leq k$. We say that an execution $e = e_1e_2\ldots e_n$ is \emph{$k$-bounded} 
from $\conf_1$ if $\conf_1 \xrightarrow{e_1} \conf_2 \ldots \conf_n \xrightarrow{e_n} \conf_{n+1}$ and for all $1 \leq i \leq n+1$, $\conf_i$ is $k$-bounded; we denote this as $\conf_1 \xrightarrow{e}_k \conf_{n+1}$.

We define the $k$-reachability set of
$\system$ to be the largest subset $\RSK{k}{\system}$ of 
$\RS{\system}$ within which each configuration $\conf$ can be 
reached by a $k$-bounded execution from $\conf_0$. 
Note that, given a FIFO system
$\system$, for every integer $k$, the set $\RSK{k}{\system}$ 
is finite and computable. 

\begin{example}\label{ex:interferences}
	Let us revisit the system in Figure~\ref{fig:ex1} and explore each of the interferences with the following executions (we denote by \textcolor{red}{red} the messages subject to interference):
	\begin{itemize}
		\item \emph{Corruption}: Let us consider execution $e_c = \actsend{\procx{s}}{\procr}{start} ~.~ \actrec{\procx{s}}{\procr}{start} ~.~ \textcolor{red}{\actsend{\procx{s}}{\procr}{data}} ~.~$ $ \textcolor{red}{\actrec{\procx{s}}{\procr}{end}} ~.~ \actsend{\procr}{\procx{s}}{ack} ~. \actsend{\procx{s}}{\procr}{data}$. Here, the message $\msg{data}$ has been corrupted to $\msg{end}$. Hence, process $\procr$ incorrectly receives the message $\msg{end}$, and assumes that process $\procx{s}$ has stopped sending $\msg{data}$, while process $\procx{s}$ continues to send it. 
		\item \emph{Lossiness}: Consider the execution $e_\ell = \actsend{\procx{s}}{\procr}{start} ~.~ \actrec{\procx{s}}{\procr}{start} ~.~ \actsend{\procx{s}}{\procr}{data} ~.~ \actrec{\procx{s}}{\procr}{data} ~.~$ $ \textcolor{red}{\actsend{\procx{s}}{\procr}{end}}$. Here, the message $\msg{end}$ has been lost, which means process $\procr$ will be stuck waiting for process $\procs$ to either send $\msg{data}$ or $\msg{end}$.
		\item \emph{Out-of-order}: Let  $e_o = \actsend{\procx{s}}{\procr}{start} ~.~ \actrec{\procx{s}}{\procr}{start} ~.~ \textcolor{red}{\actsend{\procx{s}}{\procr}{data}} ~.~  \textcolor{red}{\actsend{\procx{s}}{\procr}{end}} ~.~ \textcolor{red}{\actrec{\procx{s}}{\procr}{end}} ~.~ \actsend{\procr}{\procx{s}}{ack} ~.~ \actrec{\procr}{\procx{s}}{ack} ~.~ \textcolor{red}{\actrec{\procx{s}}{\procr}{data}}
		.$ $
		\actsend{\procr}{\procx{s}}{err} ~.~ \actrec{\procr}{\procx{s}}{err}$. In this case, the order of $\msg{data}$ and $\msg{end}$ has been swapped in the queue, which leads to a configuration where the error message is sent.
	\end{itemize}
\end{example}

As shown in \cite{abdulla_verifying_1993, finkel_decidability_1994}, for communicating automata with lossiness, the reachability set is recognisable, and the reachability problem is decidable. In the case of out-of-order scheduling, it is easy to see that the problem reduces to reachability in Petri nets. It is less clear, but it can also be reduced to Petri net reachability problem in case of corruption. We recall these proofs in Appendix~\ref{app:sec2}. 
However, the complexity of reachability for these systems is
very high -- it is non-primitive recursive for lossy
systems \cite{Schnoebelen02}, and Ackermann-hard for corruption
and out-of-order \cite{CzerwinskiLLLM21}. Hence, it is still worth exploring subclasses in the presence of errors.

\section{RSC systems with interferences}\label{sec:rsc}
We first extend the definitions of synchrony in  systems from \cite{di_giusto_towards_2021} to consider possible interferences. The main extension relates to the definition of \emph{matching pairs}. Intuitively, matching pairs refer to a send action and the corresponding receive action in a given execution. In the presence of interferences, it is not necessary that the same message that is sent is received (due to corruption), or that the $k$-th send action corresponds to the $k$-th receive action (due to lossiness or out-of-order). 
Hence we extend the definition of matching pairs.

\begin{definition}[Matching pair with interference]
	\label{def:matcherrors}
	Given an execution $e = a_1 \ldots a_n$, if there exists a channel $\chan{\procp}{\procq}$, messages $m, m' \in \msgset$ and $j, j', k, k' \in \{1, \ldots, n\}$ where $j < j'$, and the following four conditions:

{\rm (1)} $a_j = \actsend{\procp}{\procq}{m}$;
{\rm (2)} $a_{j'} = \actrec{\procp}{\procq}{m'}$;
{\rm (3)} $a_j$ is the $k$-th send action to
		$\chan{\procp}{\procq}$ in $e$; and 
{\rm (4)} $a_{j'}$ is the $k'$-th receive action on
		$\chan{\procp}{\procq}$ in $e$,  
	then we say that $\{j, j'\} \subseteq \{1, \ldots , n\}$ is a \emph{matching pair with interference}, or $i$-matching pair.
\end{definition}

\noindent Note that if $m = m'$ and $k = k'$, we are back to the original
definition of matching pairs in \cite[Section~2]{di_giusto_multiparty_2023}, which we shall refer to henceforth as \emph{perfect matching pairs}. When we refer to a matching pair, we mean either a perfect or $i$-matching pair. Moreover, our formalism allows for a single message to have more than one kind of interference, e.g. the same message can be corrupted and received out-of-order.

\begin{example}\label{ex:matchpairwint}
	Consider the following execution $e = a_1\ldots a_5 =  \actsend{\procp}{\procq}{a} \cdot \actsend{\procq}{\procp}{b} \cdot \actrec{\procq}{\procp}{b} \cdot \actsend{\procp}{\procq}{c} \cdot \textcolor{red}{\actrec{\procp}{\procq}{c}}$. For the channel $\chan{\procq}{\procp}$, we have a perfect matching pair $\{2,3\}$ which corresponds to the actions $\actsend{\procq}{\procp}{b}$ and $\actrec{\procq}{\procp}{b}$, the 1$^{\text{st}}$ send and receive action along $\chan{\procq}{\procp}$. For the channel $\chan{\procp}{\procq}$, we see that the 1$^\text{st}$ receive action is not $\actrec{\procp}{\procq}{a}$, and hence, there is no perfect matching pair corresponding to $\actsend{\procp}{\procq}{a}$. However, in case of interferences, we can have the following cases:
	\begin{itemize}
		\item If the message $a$ is lost, i.e., $\actsend{\procp}{\procq}{a}$ would be a lost action, $\actsend{\procp}{\procq}{c} \cdot \textcolor{black}{\actrec{\procp}{\procq}{c}}$ would be a matched send-receive pair, and therefore, $\{4,5\}$ would be the corresponding $i$-matching pair.
		\item If the message $a$ was corrupted to $c$, then, $\actrec{\procp}{\procq}{c}$ would be the receive action corresponding to $\actsend{\procp}{\procq}{a}$, and we would have $\{1,5\}$ as an $i$-matching pair.
		\item If the trace with an appended action as follows: $e' = \actsend{\procp}{\procq}{a} \cdot \actsend{\procq}{\procp}{b} \cdot \actrec{\procq}{\procp}{b} \cdot \actsend{\procp}{\procq}{c} \cdot \textcolor{black}{\actrec{\procp}{\procq}{c}} \cdot \textcolor{red}{\actrec{\procp}{\procq}{a}}$, then it could be that messages $a$ and $c$ were scheduled out-of-order in the channel $\chan{\procp}{\procq}$. Then we have $i$-matching pairs $\{1,6\}$ and $\{4,5\}$.   
	\end{itemize}
\end{example}
\noindent We now modify the definition of \emph{interactions} from \cite{di_giusto_multiparty_2023} as follows.

\begin{definition}[Interaction] An \emph{interaction} of  $e$ is either a (perfect or $i$-) matching pair, %
	or a singleton $\{j\}$ such that $a_j$ is a send action and $j$ does not belong to any matching pair (such an interaction is called unmatched send). %
\end{definition}
\noindent Given $e = a_1 \cdots a_n$, a set of interactions $\comsett$ is called a \emph{valid communication} of $e$ if for every index $j \in \{1, \ldots, n\}$, there exists exactly one interaction $\exch \in \comsett$ such that $j \in \exch$. I.e., we need to ensure that every action in $e$ belongs to exactly one interaction in the valid communication. We denote by $\setcomm(e)$ the set of all valid communications associated to $e$. 

\begin{example}\label{ex:validcomms}
Revisiting Example~\ref{ex:matchpairwint}, given the execution
$e = a_1\ldots a_5 =  \actsend{\procp}{\procq}{a} 
\cdot \actsend{\procq}{\procp}{b} \cdot 
\actrec{\procq}{\procp}{b} \cdot 
\actsend{\procp}{\procq}{c} \cdot 
\textcolor{black}{\actrec{\procp}{\procq}{c}}$, 
there are two valid communications, 
$\comsett_1 =  \{ \{1,5\}, \{2,3\}, \{4\} \}$ and 
$\comsett_2 = \{ \{1\}, \{2,3\}, \{4,5\}\}$, 
and $\setcomm(e) = \{\comsett_1, \comsett_2\}$. 
\end{example}

For the rest of this section, when we refer to an execution, we are referring to a tuple $(e, \comsett)$ such that $\comsett \in \setcomm(e)$. We say that two actions $a_1$, $a_2$ \emph{commute} if $\proc(a_1) \neq \proc(a_2)$ and they do not form a matching pair. 

Given an execution $(e, \comsett)$ such that $e = a_1 \ldots a_n$ and $\comsett \in \setcomm(e)$, we say that $j \prec_{e,\comsett} j'$ if (1) $j < j'$ and (2) $a_j$, $a_{j'}$ do not commute in $\comsett$.
We now graphically characterise \emph{causally equivalent} executions, using the notion of a conflict graph. This allows us to establish equivalences between different executions in which actions can be interchanged. 

\begin{definition}[Conflict graph]
	Given an execution $(e, \comsett)$, the conflict graph $\cgraph(e,\comsett)$ is the directed graph $(\comsett, \rightarrow_{e, \comsett})$ where for all interactions $\exch_1, \exch_2 \in \comsett$, $\exch_1 \rightarrow_{e, \comsett} \exch_2$ if there is $j_1 \in \exch_1$ and $j_2 \in \exch_2$ such that $j_1 \prec_{e, \comsett} j_2$.
\end{definition}

The conflict graph corresponding to Example~\ref{ex:validcomms}, 
$\cgraph(e,\comsett_1)$, is:
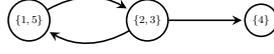
\begin{center}
		\begin{tikzpicture}[>=stealth,node distance=1cm,shorten >=1pt,
			every state/.style={text=black, scale =0.5, minimum size=0.1em}, semithick,
			font={\fontsize{8pt}{12}\selectfont}]

			\node[state] (place1) at (0, 0) {$\{1,5\}$};
			\node[state, right=1 cm of place1] (place2) {$\{2,3\}$};
			\node[state, right=1 cm of place2] (place3) {$\{4\}$};

			\path [-stealth]
			(place1) edge[bend left = 30](place2)
			(place2) edge[](place3)
			(place2) edge[bend left = 30](place1);

	\end{tikzpicture}

 \end{center}
Two executions $(e, \comsett)$ and $(e', \comsett')$ are causally equivalent, denoted by $(e, \comsett) \sim (e', \comsett')$, if their conflict graphs are isomorphic.

We are now ready to define $\iRSC$  systems, which is the extension of $\RSC$  to include interferences. Intuitively, $\iRSC$ executions can be reordered to mimic rendezvous (or synchronous) communication. In the case with interference, we enforce that every valid communication is equivalent to a RSC execution.  

\begin{definition}[\textsc{$i$-RSC}  system] An execution $(e, \comsett)$ is $\iRSC$ if all matching pairs in $\comsett$ are of the form $\{j, j+1\}$. A system $\system$ is $\iRSC$ if for all tuples $(e, \comsett)$ such that $e \in \exec(\system)$ and $\comsett \in \setcomm(e)$, we have $\cgraph(e, \comsett) = \cgraph(e', \comsett')$ where $(e', \comsett')$ is an $\iRSC$ execution.
\end{definition}

\begin{example}
\label{ex:irsc}
From Ex.~\ref{ex:validcomms}, $(e, \comsett_2)$ is an $\iRSC$ execution, but $(e, \comsett_1)$ is neither an $\iRSC$ execution nor equivalent to one, as message $\msg{a}$ has to be sent before message $\msg{b}$ is received by process $\procp$ while message $\msg{b}$ has to be sent before the corresponding message (which is now $\msg{c}$ due to corruption) is received by process $\procq$.
\end{example}

 This is the strictest version, however, this can be adapted to include only one communication by assuming a single instance of $\comsett$ instead of all. We formalise our observation about non-$\iRSC$ behaviours from Example
\ref{ex:irsc}, and show that 
$\iRSC$ still maintains the good properties of the conflict graph as in \cite{di_giusto_multiparty_2023}.

\begin{restatable}{lemma}{rscconflict}
	An execution $(e, \comsett)$ is causally equivalent to an $\iRSC$ execution iff the associated conflict graph $\cgraph(e, \comsett)$ is acyclic.

\end{restatable}

A borderline violation for interferences defined below is a key concept for the decidability of $\RSC$ systems. Intuitively, it provides a ``minimal counter-example'' for non-$\RSC$ behaviour.

\begin{definition}[Borderline violation] An execution $(e, \comsett)$ is a borderline violation if {\rm (1)} $(e, \comsett)$ is not causally equivalent to an $\iRSC$ execution, {\rm (2)} $e = e' \cdot \recact{c}{m}$ for some execution $e'$ such that {\rm (a)} for all $\comsett' \in \setcomm(e')$, $(e', \comsett')$ is equivalent to an $\iRSC$ execution and {\rm (b)} there exists $\comsett_1 \in \setcomm(e')$ such that $(e', \comsett_1)$ is an $\iRSC$ execution.
\end{definition}
\vspace{-1mm}

\begin{restatable}{lemma}{bvrsc}
$\system$ is $\iRSC$ if and only if for all $e \in \exec(\system)$ and $\comsett \in \setcomm(e)$, $(e, \comsett)$ is not a borderline violation.
\end{restatable}
Following the same approach as in \cite{di_giusto_multiparty_2023}, we show that inclusion into the $\iRSC$ class is decidable. %
For simplicity, we construct the following sets: $\actset_{nr} = \{\sract{c}{m} \mid \sendact{c}{m} \in \actset, \recact{c}{m'} \in \actset\} \cup \{\sendact{c}{m}\mid\sendact{c}{m} \in \actset\}$ and $\actset_{?} = \{\recact{c}{m}\mid\recact{c}{m} \in \actset\}$. Note that $\sract{c}{m}$ could include a send-receive pair where the message sent is different from the one received. This ensures inclusion of matching pairs due to corruption. An $\iRSC$ execution can be represented by a word in $\actset_{nr}^*$ and a borderline violation by a word in $\actset_{nr}^*.\actset_{?}$. We first show that the set of borderline violations is regular.

\begin{restatable}{lemma}{abv}
Let $\system$ with $\product(\system) = (Q, \Sigma, \chanset, \actset, \transrel, q_o)$. There is a non-deterministic finite state automaton $\aubv$ computable in time $\mathcal{O}(|\chanset|^3|\msgset|^2)$ such that $\lang(\aubv) = \{ e \in \actset_{nr}^*.\actset_? \mid \exists \comsett \in \setcomm(e) \text{ such that } (e, \comsett) \text{ is a borderline violation}\}$.

\end{restatable}
Next we show that the subset of executions $\exset(\system)$ that begin with an $\iRSC$ prefix and terminate with a reception is regular.
The construction of the automaton $\aursc$ recognising such a language mimics the $\iRSC$ executions of the original system $\system$, storing only the information on non-empty buffers, guessing which is the send message that will be matched by the final reception.

For the following result, we let the size  $|\automata|$ of an automaton $\automata = (\cstateset, \transrel, \initp{})$ be $|\cstateset| + |\transrel|$. Moreover, the size $|\system|$ of a system $\system = (\automata_{\procp})_{\procp \in \procset}$ = $\sum_{\procp \in \procset}|\automata_{\procp}|$.

\begin{restatable}{lemma}{arsc}
	Let $\system$ be a FIFO system. There exists a
        non-deterministic finite state automaton
        $\aursc$ over $\actset_{nr} \cup \actset_{?}$ such that
        $\lang(\aursc) = \{ e\cdot \recact{\chan{\procp}{\procq}}{m} \in \actset_{nr}^*.\actset_? \mid  e\cdot \recact{\chan{\procp}{\procq}}{m} \in \exec(\system)$ and  $\exists \comsett \in \setcomm(e)$ such that $(e, \comsett)$  is an $\iRSC$ execution$\}$, which can be constructed
        in time $\mathcal{O}(n^{|\procset|+2}|\chanset|^2\times
        2^{|\chanset|})$, where $n$ is the size of $\system$. 
\end{restatable}

Using the above lemmas, we derive the following main theorem in this
section, which states that the inclusion of
an $\iRSC$ system is decidable; and the complexity is comparable to that of checking inclusion to \RSC
\cite[Theorem~12]{di_giusto_multiparty_2023}.

\begin{restatable}{theorem}{rscmainthm}
\label{thm:rscmain}
Given a system $\system$ of size $n$, deciding whether it is an $\iRSC$ system can be done in time $\mathcal{O}(n^{|\procset|+2}|\chanset|^5\times
2^{|\chanset|} \times |\Sigma|^2)$.

\end{restatable}
\section{$k$-Multiparty Compatibility with interferences}\label{sec:kmc}
We extend our analyses to consider 
\emph{$k$-multiparty compatibility} ($\kmc$) which was introduced in \cite{lange_verifying_2019}
for a subset of FIFO systems, called 
\emph{communicating session automata} (CSA). 
CSA strictly include systems corresponding to asynchronous multiparty session types \cite{denielou_multiparty_2013}.

\begin{definition}[Communicating session automata]
\label{def:csa}
A deterministic FIFO automaton which has no mixed states is defined as a \emph{session automaton}. FIFO systems comprising session automata are referred to as communicating session automata (CSA).
\end{definition}
In this section, we only consider communicating session automata. We begin by recalling the definition of $\kmc$ which is composed of two properties, $k$-safety and $k$-exhaustivity. 

\begin{definition}[$k$-Safety, Definition~4 in \cite{lange_verifying_2019}]\label{def:ksafe}
A communicating system $\system$ is \emph{$k$-safe} if the following holds for all $\confq{\globq}{\chancon} \in \RSK{k}{\system}$:
 \begin{description}
 	\item[\textnormal{\textsc{($k$-er)}}] $\forall \chan{\procp}{\procq} \in \chanset$, if $
 	\chancon_{\chan{\procp}{\procq}} = m.u$ then $\confq{\globq}{\chancon} \xrightarrow{*}_k\xrightarrow{\recact{\chan{\procp}{\procq}}{m}}_k$.
 	\item [\textnormal{\textsc{($k$-pg)}}] if $q_\procp$ is receiving, then $\confq{\globq}{\chancon} \xrightarrow{}_k\xrightarrow{\recact{\chan{\procp}{\procq}}{m}}$ for some $m \in \msgset$.
 \end{description}
\end{definition}
\noindent The $k$-safety condition is composed of two properties, the first being eventual reception \textsc{($k$-er)} which ensures that every message sent to a channel is eventually received. The other property is progress \textsc{($k$-pg)} 
where the system is not ``stuck'' at any receiving state.  

A system is $k$-exhaustive if for all $k$-reachable configurations, whenever a send action is enabled, it can be fired within a $k$-bounded execution. 

\begin{definition}[$k$-Exhaustivity, Definition~8 in \cite{lange_verifying_2019}]\label{def:kexh}
	A communicating system $\system$ is $k$-exhaustive if
	for all $\confq{\globq}{\chancon} \in \RSK{k}{\system}$ and $\chan{\procp}{\procq} \in \chanset$, if $q_\procp$ is a sending state, then for all $(q_\procp, \sendact{\chan{\procp}{\procq}}{m}, q'_\procp)\in \delta_\procp$, there exists $\confq{\globq}{\chancon} \xrightarrow{*}_k \xrightarrow{\sendact{\chan{\procp}{\procq}}{m}}_k$. 
\end{definition}

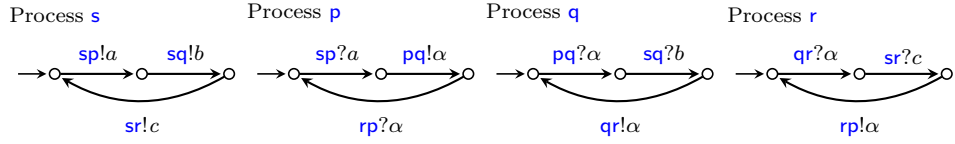
\begin{figure}
 	\vspace{-5mm}
\small{
\begin{center}
	\begin{tikzpicture}[>=stealth,node distance=1.1cm,shorten >=1pt,
		every state/.style={text=black, scale =0.45, minimum size=0.1em}, semithick,
		font={\fontsize{8pt}{12}\selectfont}]
		\node[state, initial, initial text=] (place1) at (0,0) {};
		\node[state, right=1 cm of place1] (place2) {};
		\node[state, right=1cm of place2] (place3) {};
		\node[above= 0.5cm of place1] (heading) {Process $\procx{s}$};
		
		\path [-stealth, thick]
		(place2) edge node[above=0.01cm] {$\sendact{\chan{s}{q}}{b}$}(place3)
		(place1) edge node[above=0.01cm] {$\sendact{\chan{s}{p}}{a}$} (place2)
		(place3) edge[bend left] node[below=0.1cm] {$\sendact{\chan{s}{r}}{\textit{c}}$} (place1);

		\node[state, initial, initial text=, right=3cm of place1] (place5) {};
		\node[state, right=1 cm of place5] (place6) {};
		\node[state, right=1cm of place6] (place7) {};
		\node[above=0.5cm of place5] (heading2) {Process $\procp$};
		
		\path [-stealth, thick]
		(place5) edge node[above=0.01cm]{$\recact{\chan{s}{p}}{a}$}(place6)
		(place6) edge node[above=0.01cm] {$\sendact{\chan{p}{q}}{\alpha}$} (place7)
		(place7) edge[bend left] node[below=0.1cm] {$\recact{\chan{r}{p}}{\alpha}$} (place5);

		\node[state, initial, initial text=, right=3cm of place5] (place8) {};
		\node[state, right=1 cm of place8] (place9) {};
		\node[state, right=1cm of place9] (place10) {};
		\node[above=0.5cm of place8] (heading3) {Process $\procq$};
		
		\path [-stealth, thick]
		(place8) edge node[above=0.01cm]{$\recact{\chan{p}{q}}{\alpha}$}(place9)
		(place9) edge node[above=0.01cm] {$\recact{\chan{s}{q}}{b}$} (place10)
		(place10) edge[bend left] node[below=0.1cm] {$\sendact{\chan{q}{r}}{\alpha}$} (place8);
		
		\node[state, initial, initial text=, right=3cm of place8] (place11) {};
		\node[state, right=1 cm of place11] (place12) {};
		\node[state, right=1cm of place12] (place13) {};
		\node[above=0.5cm of place11] (heading4) {Process $\procr$};
		
		\path [-stealth, thick]
		(place11) edge node[above=0.01cm]{$\recact{\chan{q}{r}}{\alpha}$}(place12)
		(place12) edge node[above=0.01cm] {$\recact{\chan{s}{r}}{c}$} (place13)
		(place13) edge[bend left] node[below=0.1cm] {$\sendact{\chan{r}{p}}{\alpha}$} (place11);
		
\end{tikzpicture}
\end{center}
}
\vspace*{-3mm}
\caption{The above system $\system$ is $k$-MC but has a non-regular
		reachability set.}
\label{fig:exkmc}
 \end{figure}

Example~\ref{ex:nonregular} shows 
that the reachability set of $\kmc$ is not necessarily regular,
unlike the binary half-duplex systems~\cite{cece_verification_2005}. 

\begin{example}
\label{ex:nonregular}
  	The system $\system$ depicted in Figure~\ref{fig:exkmc} is an example of a $\kmc$ system for which the reachability set is not regular. It
  	consists of four participants, sending messages amongst themselves. The first participant $\procx{s}$ can send equal number of $a$, $b$, $c$ letters in a loop to participants $\procp$, $\procq$, and $\procx{r}$ respectively. The participants $\procp$, $\procq$, and $\procx{r}$ behave similarly, so let us take the example of $\procp$. It consumes one letter from the channel $\chan{\procx{s}}{\procp}$, then as a way of synchronisation sends a message $\alpha$ to $\procq$ and waits to receive a message $\alpha$ from $\procx{r}$. This ensures that participants $\procp$, $\procq$, and $\procx{r}$ consume equal number of letters from their respective channels with $\procx{s}$. Hence, the reachability set for initial configuration $(s_0, p_0, q_0, r_0)$ is $a^n\#b^n\#c^n$ which is context-sensitive, hence non-regular. 
\end{example} 

\noindent We prove that $\kmc$ in the absence of errors,
for a large class of systems the $k$-safety property subsumes $k$-exhaustivity.

\begin{restatable}{theorem}{basicksafe}
\label{thm:basicksafe}
	If a directed CSA $\system$ is $k$-safe, then $\system$ is $k$-exhaustive.
\end{restatable}

\subparagraph*{$k$-MC with interferences.} 
Theorem~\ref{thm:basicksafe}
shows that the $k$-safety is a
too strong condition in the presence of interferences. For instance, in case of lossiness, progress cannot be guaranteed. This is because there is always the potential of losing messages and being in a receiving state forever.  We are now ready to define $k$-weak multiparty compatibility. %

 \begin{definition}[$k$-Weak Multiparty Compatibility]\label{def:wkmc}
A communicating system $\system$ is \emph{weakly $\kmc$}, or $\kwmc$, if it satisfies $\textsc{$k$-er}$ and is $k$-exhaustive.%
 \end{definition}

This notion covers a larger class of systems than $\kmc$ systems, and it is more natural in the presence of errors. Moreover, we still retain the decidability of $\kwmc$ in the presence of errors. We briefly discuss weaker refinements to these properties in \S~\ref{sec:conclusion}. We conclude with the following theorem which states that the $\kwmc$ property is decidable.
 
 \begin{restatable}{theorem}{kmcmainthm}
 	\label{thm:kmcmainthm}
 	Given a system $\system$ with lossiness (resp. corruption, resp. out-of-order) errors, checking the $\kwmc$ property is decidable and \textsc{pspace}-complete.
 \end{restatable}

\section{Crash-stop failures} \label{sec:crashstop}
Session types
\cite{THK,honda.vasconcelos.kubo:language-primitives,honda_multiparty_2008}
are a type discipline to ensure communication safety for message
passing systems. Most session types assume a scenario where
participants operate reliably, i.e. communication happens without
failures. To model systems closer to the real world, Barwell et
al.~\cite{barwell_generalised_2022, barwell_designing_2023} introduced
session types with \emph{crash-stop failures}. In this section, we
consider the same notion for communicating systems which we define as
crash-handling.

\subsection{Crash-handling FIFO systems}
\label{subsec:fifocs}

We extend this framework to FIFO systems. As in \cite{barwell_generalised_2022,barwell_designing_2023}, we declare a (potentially empty) set of \emph{reliable processes}, which we denote as $\relset \subseteq \procset$. 
If a process is assumed reliable, the other processes can interact with it 
without needing to handle its crashes. 
Hence if $\relset = \procset$, there is no additional crash-handling behaviour for the system. In this way, we can model a mixture of reliable and unreliable processes. For simplicity in the construction, we enforce an additional constraint that in the crash-handling branches, there is no receive action from the crashed process.

We use a shorthand for the \emph{broadcast} of 
a message $m \in \Alphset$ by process $\procp \in \procset$ 
along all outgoing channels:
$(q, \broadcast{m}{p}, q')$ 
if 
$q \xrightarrow{ \actsend{\procp}{\procx{r_1}}{m}.
	\actsend{\procp}{\procx{r_2}}{m}\ldots
	\actsend{\procp}{\procx{r_n}}{m}}
q'$ 
such that  
$n=|\chansetout|$ and $\procx{r_i} \neq \procx{r_j}$ for
all $\procx{i} \neq \procx{j}$. We denote by $\cbroadcast{m}{\procp}$ the concatenation $\intact{crash}_\procp \cdot \broadcast{m}{p}$ where $\intact{crash}_\procp \in \internalset$ is an internal action reserved for when process $\procp$ crashes.

Let $\system = (\automata_\procp)_{\procp \in \procset}$ be a FIFO system %
over $\Alphset \uplus \{\text{\lightning}\}$. Let the set of reliable processes be $\relset \subseteq \procset$. For each ${\procp \in \procset}$: \begin{itemize}
	\item We divide the state set as follows : $\cstateset_\procp = \cstateset_{\procp, 1} \uplus \cstateset_{\procp, 2} \uplus \cstateset_{\procp, 3}$.
	\item Let $\actset \subseteq (\chanset \times \{!, ?\} \times(\msgset \uplus \{\text{\lightning}\}) )
	\uplus \internalset$ be the set of actions. 
	\item We split $\transrel_\procp = \transrel_{\procp, 1} \uplus \transrel_{\procp, 2}$ such that: \begin{itemize}
		\item $\transrel_{\procp, 1} \subseteq \cstateset_{\procp,1} \times (\chanset \times \{!, ?\} \times \msgset) \times (\cstateset_{\procp,1} \cup \cstateset_{\procp,2})$, and
		\item $\transrel_{\procp, 2} \subseteq \cstateset_{\procp} \times [( \chanset \times \{!, ?\} \times \{\text{\lightning}\}) \cup \internalset] \times \cstateset_{\procp}$. 
	\end{itemize}  
\end{itemize}
 We say that a process $\procp$ has crash-handling behaviour in $\system$ if $\transrel_{\procp, 2}$ is the smallest set of transitions such that:
\begin{enumerate}
	\item \emph{Crash handling} $\CI$: For all $(q, \actrec{\procr}{\procp}{a}, q') \in \transrel_{\procp, 1}$ such that $\initconf \xrightarrow{e} \conf = \confq{\globq}{\chancon}$ and $\globq_\procp = q$ and $\procr \in \procset \setminus \relset$, there exists $q'' \in (\cstateset_{\procp,1} \cup \cstateset_{\procp,2})$ such that $(q, \actrec{\procr}{\procp}{\text{\lightning}}, q'') \in \transrel_{\procp, 2}$. %

	\item \emph{Crash broadcast} $\CP$: If $\procp \notin \relset$, then for all $q \in
	\cstateset_{\procp,1}$, there exists a crash-broadcast $(q, \cbroadcast{\text{\lightning}}{p},\qstop) \subseteq \transrel_{\procp, 2}$, for some $\qstop \in \cstateset_{\procp, 2}$ and all intermediate states belonging to $\cstateset_{\procp, 3}$.
	\item \emph{Crash redundancy} $\CR$: Finally, we have the condition that any dangling crash messages are cleaned up. For all $q \in 	\cstateset_{\procp,2}$, $(q, \actrec{\procr}{\procp}{\text{\lightning}}, q) \in \delta_{\procp,2}$.
\end{enumerate}
\noindent Condition $\CI$ enforces that every state in the system which receives from an unreliable process has a crash-handling branch, so that the receiving process is not deadlocked waiting for a message from a process that has crashed. %
Condition $\CP$ ensures that every unreliable process can non-deterministically take the internal action $\intact{crash}_\procp$ when it crashes and broadcast this information to all the other participants.
Condition $\CR$ ensures that from all states in $\cstateset_{\procp, 2}$, any dangling crash  messages are cleaned up from an (otherwise empty) channel.

\begin{definition}[Crash-handling systems]
	We say that a system $\system$ is crash-handling if every process $\procp \in \procset$ has crash-handling behaviour in $\system$.
\end{definition}

\begin{figure}[t]
\begin{center}
\begin{tabular}{ccc}
{	\small{
		\begin{tikzpicture}[>=stealth,node distance=1cm,shorten >=1pt,
			every state/.style={text=black, scale =0.5, minimum size=0.1em}, semithick,
			font={\fontsize{8pt}{12}\selectfont}]
			\node[state, initial, initial text=] (place1) at (0,0) {};
			\node[state, right=1 cm of place1] (place2) {};
			\node[state, right=1cm of place2] (place3) {};
			\node[state, below=1 cm of place2] (place4) {};
			\node[state, below=1 cm of place3] (place5) {};
			\node[above= 0.4cm of place1] (heading) {Process
$\procx{c}$};
			\path [-stealth, thick]
			(place1) edge node[above=0.01cm] {$\sendact{\chan{c}{s}}{\ms{req}}$} (place2)
			(place2) edge node[above=0.1cm] {$\recact{\chan{s}{c}}{\ms{res}}$}(place3)
			(place1) edge node[left=0.1cm] {$\intact{crash}_{\procx{c}}$} (place4)
			(place2) edge node[right=0.01cm] {$\intact{crash}_{\procx{c}}$} (place4)
			(place4) edge node[above=0.002cm] {$\sendact{\chan{c}{s}}{\text{\lightning}}$} (place5);

			(%

\node[state, initial, initial text=, right=3.25cm of place1] (place5) {};
			\node[state, right=1 cm of place5] (place6) {};
			\node[state, right=1cm of place6] (place7) {};
			\node[state, below=1cm of place6] (place8) {};
			{};
\node[above=0.4cm of place5] (heading2) {Process $\procx{s}$};
			
			\path [-stealth, thick]
			(place5) edge node[above=0.01cm] {$\recact{\chan{c}{s}}{\ms{req}}$} (place6)
			(place6) edge node[above=0.01cm] {$\sendact{\chan{s}{c}}{\ms{res}}$}(place7)
			(place5) edge node[left=0.01cm] {$\recact{\chan{c}{s}}{\text{\lightning}}$} (place8);
	\end{tikzpicture}}
}
&
\quad \quad 
&
	\small{
		\vspace{-7mm}
		\begin{tikzpicture}[>=stealth,node distance=1cm,shorten >=1pt,
			every state/.style={text=black, scale =0.7, minimum size=0.1em, align=center}, semithick,
			font={\fontsize{8pt}{12}\selectfont}]
			\node[state, initial, initial text=] (place1) at (0,0) {$\deftypep{T_0}$};%
			\node[state, right=1.2 cm of place1] (place2){$\deftypep{T_1}$}; %
			\node[state, below=0.75 cm of place1] (place4){$\deftypep{T_2}$};
		  \node[state, right= 1.2cm of place4] (place3) {$\deftypep{T_3}$};
			
			\path [-stealth, thick]
			(place1) edge node[above=0.01cm] {$\recact{\chan{B}{C}}{\msgtyped{sig}}$} (place2)
			(place2) edge[bend left] node[below=0.05cm] {$\recact{\chan{A}{C}}{\msgtyped{commit}}$} (place1)
			(place2) edge node[right=0.1cm] {$\recact{\chan{A}{C}}{\text{\lightning}}$}(place3)
			(place1) edge node[left=0.1cm] {$\recact{\chan{B}{C}}{\msgtyped{save}}$} (place4)
(place4) edge node[below=0.01cm] {$\recact{\chan{A}{C}}{\msgtyped{finish}}$} (place3)
			(place4) edge node[above=0.01cm] {$\recact{\chan{A}{C}}{\text{\lightning}}$} (place3);			
	\end{tikzpicture}}
 \end{tabular}
\end{center}
	\caption{(a) The system $\system$ (right) is crash-handling.
(b) FIFO automata (right) of the type in Example~\ref{ex:typeintype}}
\label{fig:ex2}
\label{fig:fifofromtype}
\end{figure}
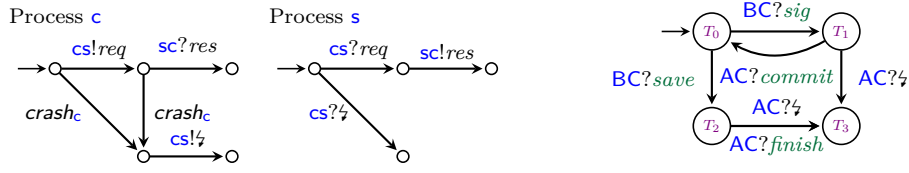

\noindent Consider the following example, which models a simple send-receive
protocol between a sender and a receiver. 

\begin{example}
	Figure~\ref{fig:ex2}(a) 
	shows a crash-handling system.
	It consists of two processes, a $\procx{s}$erver ($\procx{s}$) and a $\procx{c}$lient ($\procx{c}$). We assume that the $\procx{s}$erver is reliable, i.e. does not crash, while the $\procx{c}$lient is unreliable, i.e could crash. Hence the $\procx{c}$lient can crash in any control state, while the $\procx{s}$erver is always ready to handle a crash when it is waiting for a message from the $\procx{c}$lient.

\end{example}

In this construct, it is still possible to send messages to a crashed process. This is because from the perspective of the sending process, the crash of the receiving process is not necessarily known. Therefore, in this model, while processes can continue to send messages to crashed processes, the crashed processes would not be able to receive any messages.

Note that these properties are local to each individual automaton, hence 
the verification of these properties is decidable.

\begin{restatable}{lemma}{crashsysdec}
	\label{lem:crashsysdec}
	It is decidable to check whether a system is crash-handling.
\end{restatable}

We see that this behaviour can be appended to any FIFO system, but it does not affect the underlying verification properties of the automata. We demonstrate with the example of boundedness (i.e. checking if every execution is $k$-bounded for some $k$), but a similar argument can be used for reachability or deadlock.

\begin{restatable}{lemma}{crashreachundec}
	\label{lem:crashundecidable}
	The boundedness problem is undecidable for crash-handling systems.
\end{restatable}

\vspace{-3mm}
\subsection{Crash-handling subsystems}
Next we investigate inclusion of crash-handling systems in the aforementioned classes.

\subparagraph*{Crash-handling RSC systems.} Checking that the $\RSC$ property is decidable for crash-handling systems amounts to verifying if the proofs hold for communicating automata with internal actions. 
Let us first look at an example.

\begin{example} 	The system in Figure~\ref{fig:ex2} 
	is a crash-handling system that is also $\RSC$. We see that the behaviour of the system in the absence of crashes is $\RSC$, and in the presence of crashes, there is no additional non-$\RSC$ behaviour. Moreover, even if the $\ms{req}$ is sent, followed by the crash broadcast---since the crash message is never received, the behaviour of the system is still $\RSC$. However, this need not be the case for other examples.
\end{example}
\noindent
Next we show that 
the proofs from \cite{di_giusto_multiparty_2023} can be adapted to automata with internal actions.

\begin{restatable}{theorem}{rsccrashdec}
	\label{prop:rsccrashdec}
	Given a crash-handling system $\system$, it is decidable to
        check inclusion to the $\RSC$ class.%
\end{restatable}

\subparagraph*{Crash-handling $k$-WMC systems.}
We now show that checking $\kwmc$ is decidable for crash-handling systems generated from a collection of local types.
The reason for considering $\kwmc$ instead of $\kmc$ in 
\cite[Definition~9]{lange_verifying_2019} is that for crash-handling systems generated from
local types, the end states are receiving states (as opposed to
\emph{final states}). This result is adapted from \cite{lange_verifying_2019} with the inclusion of internal actions. %

\begin{restatable}{theorem}{crashkmcmain}
	\label{thm:kmcmain}
	Given a crash-handling system $\system$ generated from a collection of communicating session automata, it is decidable to check $\kwmc$, and can be done in $\textsc{pspace}$.
\end{restatable}
\section{Session types with crash-stop failures}
\label{sec:crashstopsessions}
This section shows that the crash-handling system
strictly subsumes the crash-stop systems in
\cite{barwell_designing_2023,BHYZ2025}, preserving the semantics.
We recall the crash-stop semantics for local types defined in \cite{barwell_designing_2023} where the major additions are (1)  a special local type 
$\stoptype$ to denote crashed processes; and  
(2) a crash-handling branch ($\crasht$) in  
one of branches of an external choice.

The syntax of the local types ($\deftypep{S}, \deftypep{T}, \ldots$) 
are given as:
\begin{align*}
	\deftypep{S}, \deftypep{T} \quad \Coloneqq& \quad \branchingtype{p}{m_i}{B_i}{T_i} ~\mid~\selectiontype{p}{m_i}{B_i}{T_i} \quad && \text{\small (external choice, internal choice)}\\
	\mid & \quad \rectypedef{T}~ \mid ~ \typet ~ \mid ~  \deftype{end} ~ \mid ~  \deftype{stop} && \text{\small (recursion, type variable, end, crash)}
\end{align*}
\noindent An external choice (branching) (resp., an internal choice (selection)), 
denoted by $\branchingtype{p}{m_i}{B_i}{T_i}$ (resp., $\selectiontype{p}{m_i}{B_i}{T_i}$) 
indicates that the \emph{current} role is to \emph{receive} from (or 
\emph{send} to) the process $\procp$.
We require pairwise-distinct,
non-empty labels and the crash-handling label ($\crasht$) 
not appear in \emph{internal} choices; 
and that singleton crash-handling labels not permitted in external choices.
The type $\deftype{end}$ 
indicates a \emph{successful} termination (omitted where
unambiguous), and recursive types are assumed \emph{guarded}, i.e., $\rectypedef{\typet}$ is not allowed, and recursive variables are unique.
A \emph{runtime} type $\deftype{stop}$ denotes crashes.

We point out here that while this is a bottom-up view of the crash-handling behaviour introduced in \cite{barwell_designing_2023}, we have taken a purely type-based approach here. For a calculus based approach, we refer the reader to \cite{barwell_designing_2023_full}.

We define the LTS over local types and extend the notions to communicating systems. We use the same labels as the ones for communicating systems.

\begin{definition}[LTS over local types]
	The relation $\deftypep{T} \xrightarrow{a} \deftypep{T'}$ for the local type of role $\procp$ is defined as:
	\vspace{-2mm}
	\begin{alignat*}{3}
		&\textnormal{\textsf{[LR1]}} \quad \sebrtype{q}{m_i}{B_i}{T_i} && \myrightarrow{\actsr{p}{q}{\msgtype{m_k}{B_i}}} && \deftypep{T_k}, \quad \text {where }\dagger \in \{!, ?\} \text{ and } \msgtyped{m_k} \neq \crasht.\\
		&\textnormal{\textsf{[LR2]}} \quad \deftypep{T}[\rectype{T}/\typet] \xrightarrow{a} \deftypep{T'} &&\xRightarrow{\mathmakebox[\arrow]{}} && \rectype{T} \xrightarrow{a} \deftypep{T'} \\
		&\textnormal{\textsf{[LR3]}} \quad \ssebrtype{q}{m_i}{B_i}{T_i} \quad &&\xrightarrow{\cbroadcast{\text{\lightning}}{p}}\quad &&  \deftype{stop} \text{, where } \sebr \in \{!, ?\}. \\
		&\textnormal{\textsf{[LR4]}} \quad \branchingtype{q}{m_i}{B_i}{T_i} && \myrightarrow{\actrec{q}{p}{\text{\lightning}}}&& \deftypep{T_k}, \quad \text{ if } \msgtyped{m_k} = \crasht. \\
		&\textnormal{\textsf{[LR5]}} \quad \qquad \deftypep{T} && \myrightarrow{\actrec{q}{p}{\text{\lightning}}}&& \deftypep{T}, \quad \forall\procq \in \procset \setminus \{\procp\} \text{ for }\deftypep{T} \in \{\deftype{stop}, \deftype{end}\}.
	\end{alignat*}

\end{definition}
\noindent Rules \textsf{[LR1]} and \textsf{[LR2]} are standard 
output/input and recursion rules, respectively; rule \textsf{[LR3]} accommodates for the crash of a process; rule \textsf{[LR4]} is the main rule for crash-handling where the reception of crash information leads the process to a crash-handling branch; and rule \textsf{[LR5]} allows any dangling crash information messages to be read in the sink states.

The LTS over a set of local types is defined as in Definition~\ref{def:fifoconf}, where a configuration $\conf = \confq{\globqt}{\chancon}$ of
a system is a pair with $\globqt = \{\deftypep{T}_\procp\}_{\procp \in \procset}$ and $\chancon = (\word_{\chan{\procp}{\procq}})_{\chan{\procp}{\procq \in \chanset}}$
with $\word_{\chan{\procp}{\procq}} \in \msgset^*$.

Next we algorithmically translate from local types
to FIFO automata preserving the trace semantics. %
Below we write $\mu\vv{\typetb}.{\deftypepb{T}}$ for $\mu{\typetb}_1.\mu{\typetb}_2\ldots\mu{\typetb}_n.\deftypepb{T}$ with $n \geq 0$.

In order to construct the FIFO automata, we first need to define the set of states. Intuitively, this is the set of types which result from any continuation of the initial local type. Below we define \emph{a type occurring in another type} (based on the definition in \cite{vassor_refinements_2024}).

\begin{definition}[Type occurring in type,
    \cite{vassor_refinements_2024}] We say a type $\deftypep{T'}$
  occurs in $\deftypep{T}$ (denoted by $\deftypep{T'} \in
  \deftypep{T}$) if and only if at least one of the following
  conditions holds: {\rm (1)} if $\deftypep{T}$ is
  $\branchingtype{p}{m_i}{B_i}{T_i}$, there exists $i \in I$ such that
  $\deftypep{T'} \in \deftypep{T_i}$; {\rm (2)} if $\deftypep{T}$ is
  $\selectiontype{p}{m_i}{B_i}{T_i}$, there exists $i \in I$ such that
  $\deftypep{T'} \in \deftypep{T_i}$; {\rm (3)} if $\deftypep{T}$ is $
  \rectypedef{T_\mu}$ then $\deftypep{T'} \in \deftypep{T_\mu}$; or
             {\rm (4)} $\deftypep{T'} = \deftypep{T}$, where $=$
             denotes the syntactic equality.
	\end{definition}

\begin{example}\label{ex:typeintype}
	Let $\procset = \{\procx{A}, \procx{B}, \procx{C}\}$ and
        $\relset = \{\procx{B}, \procx{C}\}$. Consider a local type of $\procx{C}$: 
 $\deftypep{T} = \rectypec
        \branchingtypecii{B}{sig}{B_i}{\branchingtypecii{A}{commit}{B_i}{\typetb}{\crasht}{B_i}{\deftype{end}}}{~save}{B_i}{\branchingtypecii{A}{finish}{B_i}{\deftype{end}}{{\crasht}}{B_i}{\deftype{end}}}$. 
	 
	Then, the set of all $\deftypep{T'} \in \deftypep{T}$ is $\{\deftypep{T},  \branchingtypeci{B}{sig}{B_i}{\branchingtypecii{A}{commit}{B_i}{\typetb}{{\crasht}}{B_i}{\deftype{end}}}, \branchingtypeci{A}{commit}{B_i}{\typetb}, \linebreak \branchingtypeci{B}{save}{B_i}{\branchingtypecii{A}{finish}{B_i}{\deftype{end}}{{\crasht}}{B_i}{\deftype{end}}}, \branchingtypeci{A}{{\crasht}}{B_i}{\deftype{end}},  \branchingtypeci{A}{finish}{B_i}{\deftype{end}}, \deftype{end}, \typetb\}$. 
\end{example}
And now, we are ready to define the FIFO automata.

\begin{definition}[From local types to FIFO automata]
Let $\deftypepb{T_0}$ be the local type of participant $\procp$. The automaton corresponding to $\deftypepb{T_0}$ is $\automata(\deftypepb{T_0}) = (\cstateset, \transrel, \init)$ where:\begin{enumerate}
	\item  $\cstateset = \{\deftypepb{T'} \mid \deftypepb{T'}  \in \deftypepb{T_0}, \deftypepb{T'}  \neq \typetb, \deftypepb{T'}  \neq \rectypeb{T} \} \cup \{\qsink\} \cup \{ \qsend{r} \mid \procr \in \procset \setminus \{\procp \} \} $
	
	\item $\init = \strip{\deftypepb{{T_0}}}$; 
	\item $\transrel$ is the smallest set of transitions such that $\forall \deftypep{T} \in \cstateset$:
	\begin{enumerate}
		\item If $\deftypep{T} = \sebrtype{q}{m_i}{B_i}{T_i}$ and $k \in I$, $\msgtyped{m_k} \neq {\crasht}$, and $\sebr \in \{!, ?\}$	
			\\[1mm]
			\hspace*{-0.5cm}
			$\begin{cases}
				& (\deftypep{T}, \actsr{p}{q}{\msgtype{m_k}{B_k}}, \strip{\deftypep{T_k}}) \in \transrel \quad \text{if } \deftypep{T_k} \neq \typetb \\
			
				& (\deftypep{T}, \actsr{p}{q}{\msgtype{m_k}{B_k}}, \strip{\deftypepb{T'}}) \in \transrel \quad \text{if }\deftypep{T_k} = \typetb \text{ with } \rectypeb{T'} \in \deftypep{T_0}.
			\end{cases}
			$\\	
				\item 	If $\deftypep{T} = \branchingtype{q}{m_i}{B_i}{T_i}$ 
				with $k \in I$, $\msgtype{m_k}{} = {\crasht}$
					\\[1mm]
				\hspace*{-0.5cm}
					$\begin{cases}
						& (\deftypep{T}, \actrec{q}{p}{\text{\lightning}}, \strip{\deftypep{T_k}}) \in \transrel 
						\quad \text{if } \deftypep{T_k} \neq \typetb \\
					
						& (\deftypep{T}, \actrec{q}{p}{\text{\lightning}}, \strip{\deftypep{T'}}) \in \transrel \quad \text{if }\deftypep{T_k} = \typetb \text{ with } \rectypeb{T'} \in \deftypep{T_0}.
					\end{cases}$\\

				\item If $\deftypep{T} \notin \{\deftype{stop}, \deftype{end}\}$, then $(\deftypep{T}, \cbroadcast{\text{\lightning}}{\procp}, \deftype{stop}) \subseteq \transrel$ where 
				\begin{enumerate}
					\item $(\deftypep{T}, \intact{crash}, \qsink) \in \transrel$
					\item $(\qsink, \actsend{p}{r_1}{\text{\lightning}, \qsend{r_1}}) \in \transrel$
					\item $(\qsend{r_i}, \actsend{p}{r_{i+1}}{\text{\lightning}, \qsend{r_{i+1}}}) \in \transrel$ $\forall i \in \{1, \ldots, n-2\}$, where $n = |\chansetout|$
					\item $(\qsend{r_{n-1}}, \intact{crash}, \deftype{stop}) \in \transrel$
				\end{enumerate}

				\item If $\deftypep{T} \in  \{\deftype{stop}, \deftype{end}\}$, then $(\deftypep{T}, \actrec{q}{p}{\text{\lightning}}, \deftypep{T}) \in \transrel$ for all $\procq \in \procset \setminus \{\procp\}$.
			\end{enumerate}
		\end{enumerate} 	
	\smallskip
where
$\strip{\deftypep{T}} \defeq \strip{\deftypep{T'}}$ if
$\deftypep{T} = \rectype{T'}$; otherwise
$\strip{\deftypep{T}} \defeq \deftypep{T}$.  
\end{definition}

\begin{example}
		The FIFO automata constructed from the type in Example~\ref{ex:typeintype} is shown in Figure~\ref{fig:fifofromtype}. We see that for all receiving actions from process $\procx{A}$, which is not in the reliable set of processes, there is a crash-handling branch, where:
		\vspace{-1mm}
		 {%
			\begin{align*} \deftypep{T_0} &=
		\branchingtypecii{B}{sig}{B_i}{\branchingtypecii{A}{commit}{B_i}{\typetb}{\crasht}{B_i}{\deftype{end}}}{~save}{B_i}{\branchingtypecii{A}{finish}{B_i}{\deftype{end}}{{\crasht}}{B_i}{\deftype{end}}}\\
		\deftypep{T_1} &=\branchingtypecii{A}{commit}{B_i}{\typetb}{{\crasht}}{B_i}{\deftype{end}}  \quad
		\deftypep{T_2} = \branchingtypecii{A}{finish}{B_i}{\deftype{end}}{{\crasht}}{B_i}{\deftype{end}} \quad 
		\deftypep{T_3} = \deftype{end}%
	\end{align*}%
}
	\end{example}

We now prove that the automata generated from a local type can be composed into communicating session automata, and this translation preserves the semantics.
	\begin{restatable}{lemma}{basicautomata}
\label{lem:crashcsa}
Assume $\deftypep{T}_\procp$ is a local type. Then
$\automata(\deftypep{T}_\procp)$ is
deterministic, directed and has no mixed states. Moreover, $\deftypep{T}_\procp \approx \automata(\deftypep{T}_\procp)$, i.e.  $\forall \phi, \phi \in \exec(\deftypep{T}_\procp) \Leftrightarrow \phi \in \exec{(\automata(\deftypep{T}_\procp))}$.
	\end{restatable}
	\smallskip
        \noindent By Lemma~\ref{lem:crashcsa}, we derive that
        the resulting systems belong to the class
of crash-handling systems, and the problem of checking $\RSC$ and
$\kwmc$ is decidable for this class.

	\begin{restatable}{theorem}{fifocrash}
		\label{lem:fifocrash}
		The FIFO system generated from the translation of
                crash-stop session types is a crash-handling
                system. Moreover, it is decidable to check
                inclusion to the $\RSC$ and $\kwmc$ classes. 
	\end{restatable}
\section{Experimental evaluation}\label{sec:implementation}
We verify protocols in the literature
under interferences 
in order to compare how inclusion to the $\RSC$ and $\kwmc$ classes
change; 
and which of $\RSC$ and $\kmc$ is more resilient under failures.
We used the tools, $\RSC$-checker \emph{ReSCu} 
\cite{desgeorges_rsc_2023} (implemented in OCaml), 
and $\kmc$-checker \texttt{kmc} \cite{lange_verifying_2019}
(implemented in Haskell).
The \emph{ReSCu} tool implements a version of the out-of-order scheduling with an option, so we use this available option to take our benchmark.

The  \texttt{kmc} tool implements the options
to check (1) the $k$-exhaustivity ($k$-$\textsc{exh}$, Def.~\ref{def:kexh}); 
(2) the $k$-eventual reception ($k$-$\textsc{er}$,
Def~\ref{def:ksafe}); and 
(3) the progress ($k$-$\textsc{pg}$, Def~\ref{def:ksafe}).
The $k$-weak multiparty compatibility condition ($\kwmc$,
Def~\ref{def:wkmc}) 
no longer checks for $k$-$\textsc{pg}$ but checks $k$-$\textsc{exh}$
and $k$-$\textsc{er}$. 
Hence checking (1,2,3) gives us the justification
whether $\kwmc$ is more resilient than $\kmc$. 
For the out-of-order, 
we implemented the out-of-order scheduling in Haskell
mirroring the implementation as in \emph{ReSCu}. %
To model lossiness, we add reception self-loops as defined in completely specified protocols \cite{finkel_decidability_1994}, and for corruption, we allow the sending of arbitrary messages; both of these are implemented in Python.  

Table~\ref{table:implementation} shows the evaluation results. 
From the benchmarks in \cite{desgeorges_rsc_2023}, we selected all the relevant
benchmarks (CSA and peer-to-peer) in order to evaluate them by \texttt{kmc}. 
Interestingly, in the case of out-of-order errors,
all of the protocols which satisfy $\kmc$ without errors 
still satisfy $\kmc$. 
However $\RSC$ does not satisfy some $\kmc$ protocols. 
This would imply that in most real-world examples, the flexibility
in behaviour introduced by relaxing the FIFO condition does not affect
the inclusion to $\kmc$. 
On the other hand, 
under lossiness and corruption,
most examples no longer belong to $\kmc$.
More specifically,
$k$-$\textsc{pg}$ fails for most cases.
The $k$-$\textsc{er}$ also fails
for many cases, especially in the presence of corruption.
This justifies a relevance of our
definition of $\kwmc$ under those two failures. 

\begin{center}
	
	\setlength{\tabcolsep}{1.2pt}
	\renewcommand{\arraystretch}{1.2}
	
\captionof{table}{\footnotesize Experimental evaluation of benchmarks in the \emph{ReSCu} and \texttt{kmc} tool, under the presence of no errors, out-of-order, lossiness and corruption errors. Note that the systems were checked for $k \leq 10.$ ${\color{blue}\textsc{to}}$ denotes timeout after 5 minutes. The *-marked examples originate from the ReSCu tool \cite{desgeorges_rsc_2023}, having been translated from other papers, as detailed in the original publication. \label{table:implementation}\vspace{1.5mm}} 
	\scriptsize{
		\begin{tabular}{c|c|c|c|c|ccc|c|ccc|c}
			\hline 
Protocol 
&\multicolumn{2}{c|}{~No errors~} 
&\multicolumn{2}{c|}{~Out of order~} 
&\multicolumn{4}{c|}{~Lossiness~} 
&\multicolumn{4}{c}{~Corruption~}\\
& ~\kmc~ & ~\RSC~ &  ~\kmc~ & ~\RSC~ & $k$-exh & $k$\notsotiny{-ER} 
& $k$\notsotiny{-PG} 
& ~\RSC~ & $k$-exh & $k$\notsotiny{-ER} 
& $k$\notsotiny{-PG} & ~\RSC~  \\
			\hline
			Alternating Bit \cite{PEngineering} & \Yea & \Yea & \Yea & \Yea & \Yea & \Yea & \Nay & \Yea & \Yea & \Yea & \Nay & \Yea\\
			Alternating Bit \cite{boigelot_symbolic_1996} &\Yea & \Nay  & \Yea & \Yea  & \Yea & \Yea &\Nay & \Yea & \Yea & \Yea & \Nay & \Yea \\
			Bargain \cite{LTY15} & \Yea  & \Yea  & \Yea & \Yea  & \Yea & \Yea &\Nay & \Yea & \Yea & \Nay & \Nay & \Yea \\
			Client-Server-Logger \cite{lange_verifying_2019} & \Yea & \Nay  & \Yea & \Nay & \Yea & \Yea &\Nay & \Yea & \Yea & \Yea & \Nay & \Yea \\
			Cloud System v4 \cite{GudemannSO12} & \Yea & \Yea  & \Yea & \Nay & \Nay & \Nay &\Nay & \Yea & \Nay & \Nay & \Nay & \Nay \\
			Commit protocol \cite{Bouajjani2018} & \Yea & \Yea & \Yea & \Yea & \Yea & \Nay &\Nay & \Yea & \Yea & \Nay & \Nay & \Yea \\
			Dev System \cite{PereraLG16} & \Yea & \Yea  & \Yea & \Yea  & \Yea & \Nay &\Nay & \Yea & \Yea & \Nay & \Nay & \Yea \\
			Elevator \cite{Bouajjani2018} & \Yea  & \Nay &  \Yea & \Nay & \Yea & \Yea &\Nay & \Nay & \Nay & {\color{blue}\textsc{to}} & \Nay & \Nay \\
			Elevator-dashed \cite{Bouajjani2018} & \Yea & \Nay  & \Yea & \Nay  & \Nay & \Nay &\Nay & \Nay & \Nay & {\color{blue}\textsc{to}} & \Nay & \Nay \\
			Elevator-directed \cite{Bouajjani2018} & \Yea & \Nay  & \Yea & \Nay  & \Nay & \Nay &\Nay & \Nay & \Nay & {\color{blue}\textsc{to}} & \Nay & \Nay \\
			Filter Collaboration \cite{YellinS97} & \Yea & \Yea  & \Yea & \Yea  & \Yea & \Yea &\Nay & \Yea & \Yea & \Nay & \Nay & \Yea \\
			Four Player Game \cite{LTY15} & \Yea  & \Yea &  \Yea & \Nay & \Nay & \Yea &\Nay & \Yea & \Nay & \Yea & \Nay & \Yea \\
			Health System \cite{lange_verifying_2019} & \Yea  & \Yea  & \Yea & \Yea & \Yea & \Nay &\Nay & \Yea & \Yea & \Nay & \Nay & \Yea \\
			Logistic \cite{BPMNcoreography} & \Yea  & \Yea  & \Yea & \Yea & \Yea & \Yea &\Nay & \Yea & \Nay & \Nay & \Nay & \Yea \\
			Sanitary Agency (mod) \cite{SalaunBS06} & \Yea & \Yea  & \Yea & \Yea & \Yea & \Nay &\Nay & \Yea & \Yea & {\color{blue}\textsc{to}} & \Nay & \Yea \\
			TPM Contract \cite{HalleB10} & \Yea & \Yea  & \Yea & \Nay & \Yea & \Yea &\Nay & \Yea & \Nay & \Nay & \Nay & \Nay \\
			2-Paxos 2P3A (App~\ref{sec:paxos}) & \Yea &\Yea  & \Yea & \Yea & \Yea & \Yea & \Yea & \Yea & \Yea & \Nay & \Nay & \Yea \\
			Promela I* \cite{desgeorges_rsc_2023} & \Yea &\Nay  & \Yea & \Nay & \Yea & \Yea & \Nay & \Yea & \Yea & \Yea & \Yea & \Yea \\
			Web Services* \cite{desgeorges_rsc_2023} & \Yea &\Yea  & \Yea & \Yea & \Yea & \Yea & \Nay & \Yea & \Yea & \Nay & \Nay & \Yea \\
			Trade System* \cite{desgeorges_rsc_2023} & \Yea &\Yea  & \Yea & \Yea & \Yea & \Yea & \Nay & \Yea & \Yea & \Nay & \Nay & \Yea \\
			Online Stock Broker* \cite{desgeorges_rsc_2023} & \Nay &\Nay  & \Nay & \Nay & \Nay & \Nay & \Nay & \Yea & \Nay & \Nay & \Nay & \Yea \\
			FTP* \cite{desgeorges_rsc_2023} & \Yea &\Yea  & \Yea & \Yea & \Yea & \Yea & \Nay & \Yea & \Nay & \Nay & \Nay & \Yea \\
			Client-server* \cite{desgeorges_rsc_2023} & \Yea &\Yea  & \Yea & \Yea & \Yea & \Yea & \Nay & \Yea & \Yea & \Nay & \Nay & \Yea \\
			Mars Explosion* \cite{desgeorges_rsc_2023} & \Yea &\Yea  & \Yea & \Yea & \Yea & \Nay & \Nay & \Yea & \Nay & \Nay & \Nay & \Yea \\
			Online Computer Sale* \cite{desgeorges_rsc_2023} & \Nay &\Yea  & \Nay & \Yea & \Yea & \Yea & \Nay & \Yea & \Nay & \Nay & \Nay & \Yea \\
			e-Museum* \cite{desgeorges_rsc_2023} & \Yea &\Yea  & \Yea & \Nay & \Yea & \Nay & \Nay & \Yea & \Yea & \Nay & \Nay & \Yea \\
			Vending Machine* \cite{desgeorges_rsc_2023} & \Yea &\Yea  & \Yea & \Yea & \Yea & \Yea & \Nay & \Yea & \Yea & \Nay & \Nay & \Yea \\
			Bug Report* \cite{desgeorges_rsc_2023} & \Yea &\Yea  & \Yea & \Nay & \Yea & \Yea & \Nay & \Yea & \Nay & \Nay & \Nay & \Yea \\
			Sanitary Agency* \cite{desgeorges_rsc_2023} & \Nay &\Yea  & \Nay & \Yea & \Yea & \Yea & \Nay & \Yea & \Yea & \Nay & \Nay & \Yea \\
			SSH* \cite{desgeorges_rsc_2023} & \Nay &\Yea  & \Nay & \Yea & \Yea & \Yea & \Nay & \Yea & \Yea & \Yea & \Nay & \Yea \\
			Booking System* \cite{desgeorges_rsc_2023} & \Nay &\Yea  & \Nay & \Yea & \Yea & \Yea & \Nay & \Yea & \Yea & \Nay & \Nay & \Yea \\
			Hand-crafted Example* \cite{desgeorges_rsc_2023} & \Nay &\Yea  & \Nay & \Yea & \Yea & \Nay & \Nay & \Yea & \Yea & \Nay & \Nay & \Yea \\
			\hline
			
	\end{tabular}}
\end{center}

We implement a $k$-bounded version of the Paxos protocol \cite{Lamport_2001}, a consensus algorithm that ensures agreement in a distributed system despite failures like lossiness and reordering, using a process of proposing and accepting values (c.f. Appendix~\ref{sec:paxos} for the details). This version limits retry attempts to $k$. Our implementation (for $2$ retries, $2$ proposers and $3$ acceptors) shows it is $\kmc$ and $\RSC$ both without errors, and $\kmc$ under lossiness. Since Paxos does not assume corruption, it is unsurprising that it is no longer $\kmc$ under corruption.

\section{Conclusion and further related work}\label{sec:conclusion}
In this paper, we derived decidability and complexity results
for two subclasses, $\RSC$ and $\kmc$, 
under two types of communication failures: interferences and
crash-stop failures.  
In the absence of errors, $\RSC$ systems and $\kmc$ systems are incomparable, even if we restrict the analyses to $\kmck{1}$ systems. For example, \cite[Example~4]{lange_verifying_2019} is $\kmck{1}$ but not $\RSC$. Conversely, \cite[Example~4]{di_giusto_multiparty_2023} is $\RSC$ but does not satisfy the progress condition, and hence is not $\kmc$ for any $k \in \nat$. %
Despite these distinctions, both classes aim to generalise the
concept of half-duplex communication to multiparty systems. 
This serves as our primary motivation for examining failures in a uniform way across both $\RSC$ and $\kmc$ systems.

In the interference model, 
we introduced $\iRSC$ systems,
which relax the matching-pair conditions in $\RSC$;
and $\kwmc$ which omits the progress condition
to accommodate a model with no final states.  
We proved that the inclusion problem for these relaxed 
properties remain decidable within the same complexity class
as their error-free counterparts. 
The evaluation results in
\S~\ref{sec:implementation} confirm
that relaxed systems are more resilient than the original ones. 

As the second failure model, we investigated crash-stop failures. We
defined crash-handling communicating systems which strictly include the class of local types with crash-stop failures. We also proved that both $\RSC$ and $\kmc$ properties are decidable for this class. Note that multiparty session types 
with crash-stop failures studied in \cite{barwell_generalised_2022} 
are limited to \emph{synchronous} communications. Meanwhile, the asynchronous setting in
\cite{barwell_designing_2023} 
restricts expressiveness
 to a set of local types projected from global types (which is known to be 
less expressive than those not using global types
\cite{scalas_less_2019}). Therefore, both of these systems are strictly subsumed by
our crash-handling system as proven in Theorem
\ref{lem:fifocrash}.

Integrating the $\kmc$-checker and the \emph{ReSCu} tool (with support for crash-stop failures) into the Scala toolchain of \cite{barwell_designing_2023} is a promising direction for future work, potentially enabling the verification of a broader class of programs than those considered in \cite{scalas_less_2019,barwell_designing_2023}.

Due to the need to model failures in real-world distributed systems, various failure-handling systems have been studied in the session types literature, e.g., affine session types \cite{mostrous2014Affine,DBLP:journals/pacmpl/FowlerLMD19,HLY2024}, link-failure \cite{adameit2017Session} and event-driven failures \cite{OOPSLA21FaultTolerantMPST}. Interpreting their failures into our framework would offer a uniform 
analysis of behavioural typed failure processes.

\raggedbottom
\bibliography{bibliography}

\appendix
\section{Proofs from \S \ref{sec:prelim}} \label{app:sec2}
We sketch the proofs of reachability under the presence of errors. %
The proof is well-known for FIFO systems with lossiness and out-of-order errors, but we for completeness sake, we include it here.
\myparagraph{FIFO systems with lossiness.}
As shown in \cite{cece_unreliable_1996}, for lossy systems, the reachability set is recognisable, and the reachability problem is decidable.

\begin{lemma}[\cite{cece_unreliable_1996}]
	For FIFO systems with lossiness, the reachability set is recognisable.
\end{lemma}
\begin{proof}
The proof follows from the fact that upward-closed sets are recognisable. Moreover, the complement of the reachability set of lossy FIFO systems is upwards-closed (under the subword ordering). Therefore, the reachability set is recognisable (since recognisable sets are closed under complementation).
\end{proof}

\myparagraph{FIFO systems with out-of-order errors.}
For FIFO systems with out-of-order errors, reachability is decidable. %

\begin{lemma}
	For FIFO systems with out-of-order errors, reachability is decidable.
\end{lemma}
\begin{proof}
	FIFO systems with out-of-order errors can be seen as FIFO systems with bags, or multisets. Loosely speaking, this can translate to a vector addition system with states (VASS), and \cite{CzerwinskiLLLM21} shows that reachability is Ackermann-complete for VASS.
\end{proof}

\myparagraph{FIFO systems with corruption.}
In case of corruption, the reachability problem is decidable. 

\begin{lemma}
	For FIFO systems with corruption errors, the reachability problem is decidable.
\end{lemma}
\begin{proof}
Let $\system$ be a FIFO system with corruption. Let us consider a configuration $\conf = (\globq, \chancon)$, with $\chancon = (w_{\chan{\procp}{\procq}})_{\chan{\procp}{\procq} \in \chanset}$. Without loss of generality, let $w_{\chan{\procr}{\procx{s}}} \in \Sigma^*$ be the channel contents of channel $\chan{\procr}{\procx{s}}$ such that $|w_{\chan{\procr}{\procx{s}}}| = n$. Since the channel is corrupt, we know that $\{(\globq, \chancon') \mid w'_{\chan{\procp}{\procq}} = w_{\chan{\procp}{\procq}}$ for all $\chan{\procp}{\procq} \neq \chan{\procr}{\procx{s}}$ and $w'_{\chan{\procr}{\procx{s}}} \in \Sigma^*$ and $|w'_{\chan{\procr}{\procx{s}}}| = n\} \subseteq \RS{\system}$, since the existing channel contents can be corrupted to any other word of the same length. Hence, the reachability set is the union of all such sets of configurations. In order to find out which lengths of words are reachable for each configuration, it is sufficient to modify the automata such that there is only one letter replacing all the transitions. This ensures that we correctly count the length of all the words that are reachable. This translates to checking the reachability of a VASS (since each channel can now be seen as a counter without zero tests), and then, once we know if a word of length $n$ is reachable, we can be sure that any word of length $n$ is reachable from the initial state. Hence, the problem reduces to the reachability problem in VASS. 
\end{proof}

\section{Proofs from \S \ref{sec:rsc}}

\rscconflict*
\begin{proof}
	The left to right implication follows from two observations: first, two causally equivalent executions have
	isomorphic conflict graphs. Secondly, the conflict graph of an RSC execution is acyclic, because for an RSC execution and vertices $\exch_1, \exch_2$ in the conflict graph, $\exch_1 \rightarrow_{e, \comsett} \exch_2$ if there is $j_1 \in \exch_1$ and $j_2 \in \exch_2$ such that $j_1 \prec_{e, \comsett} j_2$. Moreover, if there is more than action in either $\exch_1$ or $\exch_2$, for $\iRSC$ executions by definition, $min(\exch_1) \prec_{e, \comsett} min(\exch_2)$. Therefore, if there is a cycle in the conflict graph, then this would imply $min(\exch_2) \prec_{e, \comsett} min(\exch_1)$, which would be a contradiction.
	
	For the converse direction, let us assume that a conflict graph associated to $e = a_1 a_2 \ldots a_n$ is acyclic. Let us consider the associated communication set $\comsett$. Let $\exch_1 \ll \cdots \ll \exch_n$ be a topological order on $\comsett$. Let $e' = \exch_1 \cdots \exch_n$ be the corresponding RSC execution, and $\comsett'$ the communication set associated to $e'$ that is RSC. 
	
	Let $\sigma$ be the permutation such that $e' = a_{\sigma(1)} \cdots a_{\sigma(n)}$. Following the proof idea in \cite{di_giusto_multiparty_2023}, we show that $e$ is causally equivalent to  $e'$. Let $j,j'$ be two indices of $e$, and let us show that $j \prec j'$ iff $\sigma(j) \prec \sigma(j')$. 
	
	We have that $\{j ,j'\}$ is a matching pair in $e$, iff, by construction, $\{\sigma(j), \sigma(j')\}$ is a matching pair in $e'$. If $\{j, j'\}$ is not a matching pair of $e$, then let $\exch$ and $\exch'$ be the interactions containing $j$ and $j'$ respectively. Since $j \prec j'$, there is an arrow between $\exch$ and $\exch'$ in the conflict graph, and moreover $a_j$ and $a_{j'}$ cannot commute. Note that there is an arrow in the conflict graph of $e'$ as well. Since the conflict graph is acyclic, we have $\sigma(j) \prec \sigma(j')$.

\end{proof}

\bvrsc*
\begin{proof}
	By definition, if there exists an execution $(e, \comsett)$ in system $\system$ such that it is a borderline violation, then $\system$ is not $\iRSC$.
	Conversely, if $\system$ is not $\iRSC$, let $(e, \comsett)$ be (one of) the shortest execution that is not causally equivalent to an $\iRSC$ execution. Then, $e = e' \cdot a$ such that for all $\comsett' \in \setcomm(e')$, we have $(e', \comsett')$ is equivalent to an $\iRSC$ execution. Let $\comsett''$ be the communication set of $e'$ such that it is a subset of the communication set $\comsett$. Let $(e'', \comsett'')$ be the $\iRSC$ execution that is causally equivalent to $(e', \comsett'')$. Then, there exists an execution $\hat{e}$ such that $(\hat{e}, \comsett)$ is an execution of $\system$. Moreover, if $a$ is a send action, then $(\hat{e}, \comsett)$ is $\iRSC$ which is a contradiction. Therefore, $({e}, \comsett)$ is a borderline violation.
\end{proof}

\abv*
\begin{proof}
	Let $\A_\bv = (Q_\bv, \transrel_\bv, q_{0,\bv}, \{q_f\})$, with $Q_\bv = \{q_{0,\bv}, q_f\} \cup (\chanset \times \actset \times \{0, 1\})$, and for all $a, a' \in \actset_{nr}$, for all $c \in \chanset$, $m,m' \in \Sigma$:
	\begin{enumerate}
		\item $(q_{0,\bv}, a, q_{0,\bv})$: this is a loop on the initial state that reads all interactions until the chosen send message
		\item $(q_{0,\bv}, \actsend{\procp}{\procq}{m}, (\chan{\procp}{\procq}, \actsend{\procp}{\procq}{m}, 0))$: this is the transition where we non-deterministically select the send message that is matched to the final reception to be borderline.
		\item In case we do not consider out-of-order errors, we add the following step: $((\procx{c}, a, 0), a', (\procx{c}, a, 0))$, if $\ch(a') \neq \procx{c}$: again loop for every communication but we do not accept any further communication on the channel $\procx{c}$ in order to stay borderline. Note that this step is skipped if we consider the general case with out-of-order errors as we can have matched pairs between a matched send and receive action. %
		\item $((\procx{c}, a, 0), a', (\procx{c}, a', 1))$, if $\proc(a) \cap \proc(a') \neq \emptyset$: here, the second interaction that will take part in the conflict graph cycle is guessed. We ensure there is a process in common with $a$ for there to be an edge between them.
		\item $((\procx{c}, a, 1), a', (\procx{c}, a, 1))$, if $\ch(a') \neq \procx{c}$: once again a loop for every interaction.
		\item  $((\procx{c}, a, 1), a', (\procx{c}, a', 1))$, if $\proc(a) \cap \proc(a') \neq \emptyset$: the next vertex (or vertices) (if any) of the conflict graph is guessed.
		\item $((\procx{c}, a, 1), \actrec{\procp}{\procq}{m'}, q_f)$, if $\proc(a) \cap \proc(\actrec{\procp}{\procq}{m'}) \neq \emptyset$: finally, an execution is accepted if it closes the cycle.
	\end{enumerate}
	Moreover, each transition of $\A_\bv$ can be constructed in constant time, so $\A_\bv$ can be constructed in time $\mathcal{O}(|\chanset|^3|\msgset|^2)$.
\end{proof}

\arsc*
\begin{proof}
	Let $\aursc = (Q_\rsc, \transrel_\rsc, q_{0,\rsc}, \{q_f\})$ be the non-deterministic automata, with $Q_\rsc = Q \times (\{\varepsilon\}\cup \chanset) \times 2^\chanset \cup \{q_f\}$. We define the transitions as follows:
	\begin{enumerate}
		\item First, while performing the action $a \in \actset_{nr}$, $(q, \chi, S) \xrightarrow{a} (q', \chi', S')$ if \begin{itemize}
			\item $(q, v) \implies (q',v')$ in the underlying transition system, for some buffer values $v,v'$ and for all $\procx{c} \in \chanset$, $v_c \neq \emptyset$ iff $\procx{c} \in S$ and $v'_{\procx{c}} \neq \emptyset$ iff $\procx{c} \in S'$, and 
			\item this condition is added in the absence of out-of-order errors: if a = $\chan{\procp}{\procq}!?m'$, then $\chan{\procp}{\procq}\notin S$, and
			\item either $\chi = \chi'$, or $a = \procx{c}!m$ and $\chi' = \procx{c}$
		\end{itemize}
		\item Second, while performing the action $a = \actrec{\procp}{\procq}{m}$, we have $(q, \chi, S) \xrightarrow{a} q_f$ if $\chi = \chan{\procp}{\procq}$ and $(q, a, q') \in \transrel_\system$ for some $q'$. 
	\end{enumerate}
	Each transition of $\aursc$ can be constructed in constant
	time. An upper bound on the number of transitions can be computed as follows: if $(q, \chi, S) \xrightarrow{a} (q', \chi', S')$ is a transition, then $q$ and $q'$ only differ on at most two machines (the one that executed the send, and the one that executed the receive),
	so there are at most $n^2$ different possibilities for $q'$
	once $q$ and $a$ are ﬁxed. There are at most two possibilities
	for $\chi'$ once $\chi$ and $a$ are ﬁxed, and $S'$ is fully
	determined by $S$ and $a$. Finally, there are
	$n^{|\procset|}(1+|\chanset|)\times 2^{|\chanset|} \times 2 \times |\chanset|$ possibilities for a choice of the pair $((q, \chi, S), a)$.
\end{proof}

\rscmainthm*
\begin{proof}
	The set of borderline violations of a system $\system$ can be expressed as $\lang(\aursc)\cdot\actset_{nr} \cap \lang(\aubv)$. Therefore, checking for inclusion in $\iRSC$ reduces to checking the emptiness of this intersection, which can be done in time $\mathcal{O}(n)$.
\end{proof}

\section{Proofs from \S \ref{sec:kmc}}
\basicksafe*
\begin{proof}
	We prove this by contradiction. Let us assume that $\system$ is not $k$-exhaustive. In other words, there exists $s \in \RSK{k}{\system}$ and $\chan{\procp}{\procq} \in \chanset$, such that $q_\procp$ is a sending state and there is no execution of the kind $s \xrightarrow{*}_k \xrightarrow{\sendact{\chan{\procp}{\procq}}{m}}_k$. In other words, the channel $\chan{\procp}{\procq}$ has $k$ messages already, i.e. $|w_{\chan{\procp}{\procq}}| = k$. However, since $s \in \RSK{k}{\system}$, and $\system$ is $k$-safe, and more specifically satisfies eventual reception, there exists a configuration $t \in \RSK{k}{\system}$ such that $s \xrightarrow{*}_k\xrightarrow{\recact{\chan{\procp}{\procq}}{m'}}_k t$, such that $w_{\chan{\procp}{\procq}} = m'\cdot u$. Moreover, since the execution $s \xrightarrow{e}_k t$ is $k$-bounded, we can be sure that there have been no new sends along the channel $\chan{\procp}{\procq}$. Furthermore, since $\system$ is directed and has no mixed states, in configuration $t$, $\procp$ is still at state $q_\procp$. Therefore, we now have $t \xrightarrow{\sendact{\chan{\procp}{\procq}}{m}}_k$, which contradicts our initial assumption.
\end{proof}

\kmcmainthm*
\begin{proof}
To check whether $\system$ is not $k$-exhaustive, i.e., for each sending state $q_\procp$ and send action from $q_\procp$, we check whether there is a reachable configuration from which this send action cannot be fired. Hence, we need to search $\RSK{k}{\system}$, which has an exponential number of configurations (wrt. $k$). Note that due to interferences, each of these configurations can now have modified channel contents. We need to store at most $|\procset|^n|\chanset||\Sigma|^k$ configurations, where $n$ is the maximum number of local states of a FIFO automata, following ideas from \cite{lange_verifying_2019} and \cite{bollig_propositional_2010}. Hence, the problem can be decided in polynomial space when $k$ is given in unary.

Next, to show that $\kwer$ is decidable, we check for every such reachable configuration, that there exists a receive action from the same channel (note that we do not need to ensure it is the same message). %

\end{proof}

\section{Proofs from \S \ref{sec:crashstop}}

\crashsysdec*
\begin{proof}
	To check whether a system $\system$ is crash-handling, we need to check the two properties:
	\begin{enumerate}
		\item Checking if $\CI$ is satisfied amounts to checking if for all receiving transitions $\tau \in \transrel_{\procp, 1}$ such that the sending process is in $\procset \setminus \relset$, there exists a transition $\tau' \in \transrel_{\procp, 2}$ that handles the crash.
		\item To check $\CP$, we need to check every state in a process in $\procset \setminus \relset$ and ensure it can send crash messages when it crashes.
	\end{enumerate}
	Both of these are structural checks made on the graph of the automata, hence, checking this is decidable.
\end{proof}

\crashreachundec*
\begin{proof}
	Since every FIFO system is a crash-handling process under the condition that $\procset = \relset$, this lemma is trivially true. Moreover, every FIFO system $\system$ with $\relset \subsetneq \procset$ can be translated to a crash-handling system $\system'$ such that $\system$ is $k$-bounded iff $\system'$ is at most $k+1$-bounded. %
	This can be done by adding a new sink state $q_{\textsf{sink}}$ such that for all receiving transitions $(q, \procx{c}?a, q') \in \transrel_{\procp, 1}$, we add to $\transrel_{\procp, 2}$ a transition $(q, \procx{c}?\text{\lightning}, q_{\textsf{sink}})$. Hence, $\CI$ will be handled. For enforcing $\CP$, we add to each state $q$ of an unreliable process $\procp$ the following transition $(q, \cbroadcast{\text{\lightning}}{\procp}, q_{\textsf{sink}})$. Hence, both conditions are satisfied, and these additional transitions do not add any  unboundedness to the channels (and at most one message extra to each channel). Moreover, if the original system is unbounded, then the same execution would be enabled in $\system'$.  Hence, $\system$ is $k$-bounded iff $\system'$ is at most $k+1$-bounded. %
\end{proof}

\rsccrashdec*
\begin{proof}
	This amounts to checking the $\RSC$ property in automata with internal actions. Intuitively, this amounts to ``skipping" the internal actions in the respective NFAs. In order to prove this, we let $\actset_{nr} =\{\sract{c}{m} \mid \sendact{c}{m} \in \actset, \recact{c}{m} \in \actset\} \cup \{\sendact{c}{m}\mid\sendact{c}{m} \in \actset\} \cup \internalset $. Then we follow the construction as before. Note that since internal actions do not have a channel associated to them, we do not need to make any further changes. In the conflict graph, they are considered as nodes with only process edges between them, hence, do not form cycles and can be ignored.
\end{proof}

\crashkmcmain*
\begin{proof}
	First, we observe that for any $k \in \nat$, $RS_k(\system)$ and $\xrightarrow{}_k$ are finite. Moreover, there are at most $n\cdot |\procset|$ control states in the system, where $n = max(\{|\cstateset_\procp| \mid \procp \in \procset\}$. 
	
	\textbf{$k$-exhaustivity:} We check whether $\system$ is not $k$-exhaustive, i.e., for each sending state $q_\procp$ and send action from $q_\procp$ , we check whether there is a reachable configuration from which this send action cannot be fired. The presence of internal actions and the absence of final states does not alter this proof.
	
	\textbf{eventual reception:} For each receiving state $q_\procp$, we check whether there is a reachable configuration from which one receive action
	of $\procp$ is enabled, followed by a send action that matches another receive. We proceed as in the case for $k$-exhaustivity with additional space to remember whether we are looking for the receiving state or for a matching send action.
	Note that the presence of the internal actions does not affect this property either. This is because they do not modify the channel bounds (and hence, are not bounded by $k$), and do not increase the size of $\RSK{k}{\system}$.
	
	 Therefore, the proofs can directly follow from \cite[Theorem 2]{lange_verifying_2019}.%
\end{proof}
\section{Proofs from \S \ref{sec:crashstopsessions}}

Before we construct the resulting FIFO automata, we first need to show that the set of states is finite. 

\begin{lemma}%
	Given a local type $\deftypep{T}$, the set $\{\deftypep{T'} \mid \deftypep{T'} \in \deftypep{T}\}$ is finite.
\end{lemma}
\begin{proof}
	Let us consider each of the four conditions to build the set $\{\deftypep{T'} \mid \deftypep{T'} \in \deftypep{T}\}$ from $\deftypep{T}$. In cases (1), (2) and (4), we see that $\deftypep{T'}$ is a strict prefix of $\deftypep{T}$. Moreover, in case (3), we do not add any element to the set. Therefore, since the length of $\deftypep{T}$ is finite, the set $\{\deftypep{T'} \mid \deftypep{T'} \in \deftypep{T}\}$ is finite.
\end{proof}
 Hence, we can conclude that the automata constructed from $\deftypep{T}$ has finitely many states. 

\basicautomata*
\begin{proof}
	For the determinism, we note that all $\msgtype{m_i}{}$ in $\selectiontype{q}{m_i}{B_i}{T_i}$ and $\branchingtype{q}{m_i}{B_i}{T_i}$ are distinct. Apart from this, there is only a unique local action that can be taken from a state in case of crash. Therefore, the automaton is deterministic. Directedness is by the syntax of branching and selection types, and the fact that the internal action crash leads to a state without interacting with any other participant. The message broadcast, although not explicit, can be viewed as a sequence of send transitions, thereby making the system directed. Finally, for the absence of mixed states, we can check a state is either sending or receiving state as one state represents either branching and selection type, along with $\deftype{stop}$, $\deftype{end}$ which are receiving states, and all the intermediate states (between a crash until the $\deftype{stop}$ state) are sending states. 
	
	We now show that the translation preserves the semantics.
	We show that $\deftypep{T} \xrightarrow{\tau} \deftypep{T'}$ if $\strip{\deftypep{T}} \xrightarrow{\tau} \strip{\deftypep{T'}}$. 
	Base case: Considering a transition of size 0, it trivially holds as $\init = \strip{\deftypep{T_0}}$.
	Let us assume it holds for a transition of size $k$. Now, let us consider a single transition from $\deftypep{T} \xrightarrow{a} \deftypep{T'}$.
	If the transition $a$ belongs to [LR1] or [LR2] (resp. [LR5]) that leads to $\deftypep{T'}$, then there exists a transition in 3a (resp. 3d) that leads to $\strip{\deftypep{T'}}$. Similarly, the correspondence holds for transitions from [LR4] to 3b. Note that rule [LR2] is implicitly applied because of the $\strip{}$ function. Finally, [LR3] corresponds to 3c.
	
	The reverse direction follows as above. The only change is for the crash-handling behaviour. Here, we modify the condition as follows: if $\strip{\deftypep{T}} \xrightarrow{\tau} q $, then there exists $\deftypep{T} \xrightarrow{\tau'} \deftypep{T'}$ such that there is a unique, deterministic sequence of transitions $\tau''$ such that $q \xrightarrow{\tau''} \strip{\deftypep{T'}}$. For all cases except 3c, the proof above can be adapted (with $q = \deftype{T'}$ and $\tau'' = \emptyword$). For 3c. we see that for every sequence of transitions taken, there is a unique continuation that leads to $\deftype{stop}$, and the concatenation of $\tau.\tau'' = \tau$ and leads to $\deftype{stop}$.

\end{proof}

\fifocrash*
\begin{proof}
	This can be seen by assuming $\cstateset_{\procp, 1} = \{\deftypepb{T'} \mid \deftypepb{T'}  \in \deftypepb{T_0}, \deftypepb{T'}  \neq \typetb, \deftypepb{T'}  \neq \rectypeb{T}\} \setminus \{\deftype{end}, \deftype{stop}\}$. Moreover, $\cstateset_{\procp, 2} =\{ \deftype{stop}, \deftype{end}\}$ and $\cstateset_{\procp, 3} = \{\qsink\} \cup \{ \qsend{r} \mid \procr \in \procset \setminus \{\procp \} \}$. Moreover, rules in 3a correspond to $\transrel_{\procp, 1}$ and 3b $\CI$, 3c $\CP$, 3d $\CR$ constitute $\transrel_{\procp, 2}$. With this correspondence, we see that all the conditions are satisfied, hence, it is a crash-handling system.
	From Lemmas~\ref{lem:crashcsa} and the above result, we see that the FIFO system generated from the translation of crash-stop local types is a crash-handling system and a collection of communicating session automata. Moreover, from Theorems~\ref{prop:rsccrashdec} and \ref{thm:kmcmain}, it is decidable to check inclusion into the $\RSC$ and $\kwmc$ classes. Therefore, we can check the inclusion for collection of local types generated from crash-stop session types.
\end{proof} %
\section{Paxos Protocol}\label{sec:paxos}
In this section, we implement a basic version of the single-decree Paxos protocol, which was originally described in \cite{Lamport_2001}. The Paxos algorithm has been used to implement a fault-tolerant distributed system, which is essentially a consensus algorithm aimed to ensure that network agents can agree on a single proposed value. We model this protocol using FIFO systems with faulty channels, and go on to explore if it belongs to any of the above-mentioned classes of communicating systems.

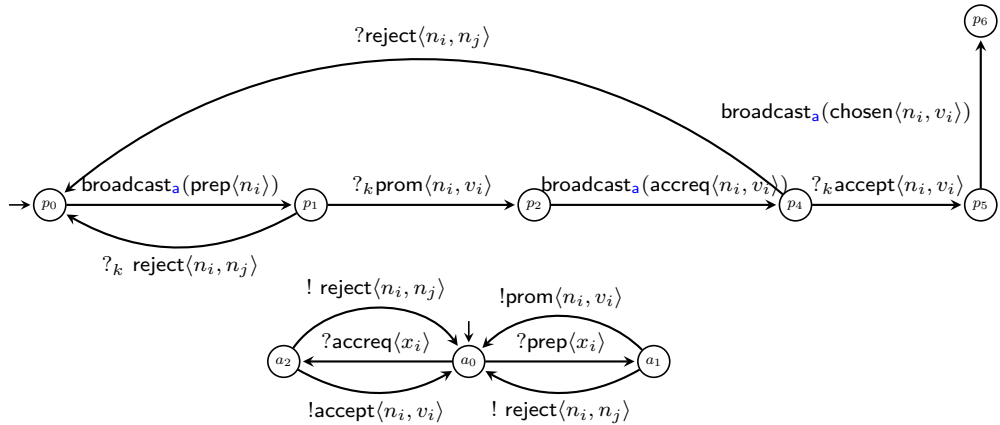
\begin{figure}
{\scriptsize
  \hspace*{-1cm}
	\begin{center}
	\begin{tikzpicture}[>=stealth,node distance=1cm,shorten >=1pt,
		every state/.style={text=black, scale =0.7, minimum size=0.1em, align=center}, semithick,
		font={\fontsize{8pt}{12}\selectfont}]
	\node[state, initial, initial text=] (place1) at (0,0) {${p_0}$};
	\node[state, right=3 cm of place1] (place2){${p_1}$};
	\node[state, right=2.5 cm of place2] (place3){${p_2}$};
	\node[state, right = 3cm of place3] (place5) {${p_4}$};
	\node[state, right=2 cm of place5] (place6){${p_5}$};
	\node[state, above=2 cm of place6] (place7){${p_6}$};
	
	\node[state, initial above, initial text=, below right= 2.5cm of place2] (plac1){${a_0}$};
	\node[state, right=2 cm of plac1] (plac2){${a_1}$};
	\node[state, left=2 cm of plac1] (plac3){${a_2}$};

	\path [-stealth, thick]
	(place1) edge node[above=0.01cm] {$\broadcast{\textsf{prep}\langle n_i\rangle}{a}$} (place2)
	(place2) edge[bend left] node[below=0.001cm] {$?_k \textsf{ reject} \langle n_i, n_j\rangle$} (place1)
	(place2) edge node[above=0.01cm] {$?_k \textsf{prom} \langle n_i, v_i\rangle$}(place3)
	(place3) edge node[above=0.01cm] {$\broadcast{\textsf{accreq}\langle n_i, v_i\rangle}{a}$} (place5)
	(place5) edge[bend right = 40] node[above=0.05cm] {$? \textsf{reject} \langle n_i, n_j\rangle$} (place1)
	(place5) edge node[above=0.01cm] {$?_k \textsf{accept} \langle n_i, v_i\rangle$}(place6)
	(place6) edge node[left=0.01cm] {$\broadcast{\textsf{chosen}\langle n_i, v_i\rangle}{a}$} (place7)
	
	(plac1) edge node[above=0.001cm] {$?\textsf{prep}\langle x_i\rangle$} (plac2)
	(plac2) edge[bend left] node[below=0.001cm] {$! \textsf{ reject} \langle n_i, n_j\rangle$} (plac1)
	(plac2) edge[bend right = 45 ] node[above=0.01cm] {$! \textsf{prom} \langle n_i, v_i\rangle$}(plac1)
	(plac1) edge node[above=0.001cm] {$?\textsf{accreq}\langle x_i\rangle$} (plac3)
	(plac3) edge[bend left = 55] node[above=0.001cm] {$! \textsf{ reject} \langle n_i, n_j\rangle$} (plac1)
	(plac3) edge[bend right] node[below=0.01cm] {$! \textsf{accept} \langle n_i, v_i\rangle$}(plac1)
	;			
	\end{tikzpicture}
	\caption{FIFO automata of the proposer and acceptor respectively in Paxos ($\broadcast{m}{a}$ refers to broadcasting message $m$ to all acceptors, $?_k m$ refers to receiving at least $k$ messages $m$).
		\label{fig:paxos}
	}
	
\end{center}

}
\end{figure} 

\paragraph*{The protocol.} We assume a subset of processes to be \emph{proposers}, i.e. processes that can propose values. The consensus algorithm ensures that exactly one value among the proposed values is chosen. A correct implementation of the protocol must ensure that:
\begin{enumerate}[\textbullet]
	\item Only a value that has been proposed will be chosen.
	\item Only one single value is chosen by the network.
	\item A process never knows that a value has been chosen unless the value has actually been chosen.
\end{enumerate}

The Paxos setting assumes the customary asynchronous, non-Byzantine model, in which:
\begin{enumerate}[\textbullet]
	\item Agents operate at arbitrary speed, may fail by crashing, and may restart. However, it is assumed that agents maintain persistent storage that survives crashes.
	\item Messages can take arbitrarily long to be delivered, can be duplicated, and can be lost or delivered out of order, but they are not corrupted.
\end{enumerate}

Paxos agents implement three roles: i) a proposer agent proposes values towards the
network for reaching consensus; ii) an acceptor accepts a value from those proposed, whereas
a majority of acceptors accepting the same value implies consensus and signifies protocol
termination; and iii) a learner discovers the chosen consensus value. 

The implementation of
the protocol may proceed over several rounds. A successful round has two phases: Prepare
and Accept.
The protocol ensures that in the case where a consensus value $v$ has already been chosen
among the majority of the network agents, broadcasting a new proposal request with a higher
proposal number will result in choosing the already chosen consensus value $v$. Following this
fact, we assume for simplicity that a learner has the same implementation as a proposer.

\begin{itemize}
	\item Requirement 1: An acceptor must accept the first proposal that it receives, i.e. for the acceptor, there must be a path from the initial state, which accepts the first proposal it gets.
	\item Requirement 2: If multiple proposals are chosen, they all have the same proposal value.
\end{itemize}

We model a bounded-version of Paxos with a FIFO system which can have any of the above-mentioned errors except corruption. The automata in Fig~\ref{fig:paxos} show an example implementation of a proposer and an acceptor, with a majority of $k$ agents needed for consensus. The action $!!\textsf{msg}$ refers to broadcasting the message $\textsf{msg}$ across all channels, and $?_k\textsf{msg}$ refers to receiving $k$ $\textsf{msg}$ messages. Both these actions can be unrolled and expressed as a combination of simple actions. Moreover, we assume that for each value of $n_i$ there is a copy of the same set of transitions. And since we cannot compare values in finite automata, we sequentially order the automata with increasing values of $n_i$. Note that this model is a CSA, and hence, we can test the $\texttt{kmc}$ tool and the ReSCu tool on an implementation. We verify that a $2$-bounded Paxos with $2$ proposers and $3$ acceptors ($2$-Paxos$2$P$3$A) is $\kmc$ and $\RSC$ in the presence of lossiness.

\end{document}